\documentclass[lettersize,journal]{IEEEtran}

\usepackage{array}
\usepackage[caption=false,font=small,labelfont=sf,textfont=sf]{subfig}
\usepackage{stfloats}
\usepackage{verbatim}
\usepackage{setspace}

\usepackage{cite}
\usepackage{amsmath,amssymb,amsfonts}
\usepackage{algorithmic}
\usepackage{graphicx}
\usepackage{textcomp}
\usepackage{xcolor}
\usepackage{graphicx}
\usepackage{float}
\usepackage{overpic}
\usepackage{comment}
\usepackage{url}
\usepackage{booktabs}

\usepackage{subeqnarray}
\usepackage{cases}
\usepackage{algorithm}

\def\BibTeX{{\rm B\kern-.05em{\sc i\kern-.025em b}\kern-.08em
		T\kern-.1667em\lower.7ex\hbox{E}\kern-.125emX}}
\begin{document}
	\newtheorem{definition}{\it Definition}
	\newtheorem{theorem}{\bf Theorem}
	\newtheorem{lemma}{\it Lemma}
	\newtheorem{corollary}{\it Corollary}
	\newtheorem{remark}{\it Remark}
	\newtheorem{example}{\it Example}
	\newtheorem{case}{\bf Case Study}
	\newtheorem{assumption}{\it Assumption}
	\newtheorem{property}{\it Property}
	\newtheorem{proposition}{\it Proposition}
	\newtheorem{observation}{\it Observation}
	
	\newcommand{\hP}[1]{{\boldsymbol h}_{{#1}{\bullet}}}
	\newcommand{\hS}[1]{{\boldsymbol h}_{{\bullet}{#1}}}
	
	\newcommand{\ba}{\boldsymbol{a}}
	\newcommand{\baq}{\overline{q}}
	\newcommand{\bA}{\boldsymbol{A}}
	\newcommand{\bb}{\boldsymbol{b}}
	\newcommand{\bB}{\boldsymbol{B}}
	\newcommand{\bc}{\boldsymbol{c}}
	\newcommand{\bcO}{\boldsymbol{\cal O}}
	\newcommand{\bh}{\boldsymbol{h}}
	\newcommand{\bH}{\boldsymbol{H}}
	\newcommand{\bl}{\boldsymbol{l}}
	\newcommand{\bm}{\boldsymbol{m}}
	\newcommand{\bn}{\boldsymbol{n}}
	\newcommand{\bo}{\boldsymbol{o}}
	\newcommand{\bO}{\boldsymbol{O}}
	\newcommand{\bp}{\boldsymbol{p}}
	\newcommand{\bq}{\boldsymbol{q}}
	\newcommand{\bR}{\boldsymbol{R}}
	\newcommand{\bs}{\boldsymbol{s}}
	\newcommand{\bS}{\boldsymbol{S}}
	\newcommand{\bT}{\boldsymbol{T}}
	\newcommand{\bw}{\boldsymbol{w}}
	
	\newcommand{\balpha}{\boldsymbol{\alpha}}
	\newcommand{\bbeta}{\boldsymbol{\beta}}
	\newcommand{\bOmega}{\boldsymbol{\Omega}}
	\newcommand{\bTheta}{\boldsymbol{\Theta}}
	\newcommand{\bphi}{\boldsymbol{\phi}}
	\newcommand{\btheta}{\boldsymbol{\theta}}
	\newcommand{\bvarpi}{\boldsymbol{\varpi}}
	\newcommand{\bpi}{\boldsymbol{\pi}}
	\newcommand{\bpsi}{\boldsymbol{\psi}}
	\newcommand{\bxi}{\boldsymbol{\xi}}
	\newcommand{\bx}{\boldsymbol{x}}
	\newcommand{\by}{\boldsymbol{y}}
	
	\newcommand{\cA}{{\cal A}}
	\newcommand{\bcA}{\boldsymbol{\cal A}}
	\newcommand{\cB}{{\cal B}}
	\newcommand{\cE}{{\cal E}}
	\newcommand{\cG}{{\cal G}}
	\newcommand{\cH}{{\cal H}}
	\newcommand{\bcH}{\boldsymbol {\cal H}}
	\newcommand{\cK}{{\cal K}}
	\newcommand{\cO}{{\cal O}}
	\newcommand{\cR}{{\cal R}}
	\newcommand{\cS}{{\cal S}}
	\newcommand{\dcS}{\ddot{{\cal S}}}
	\newcommand{\ds}{\ddot{{s}}}
	\newcommand{\cT}{{\cal T}}
	\newcommand{\cU}{{\cal U}}
	\newcommand{\wt}[1]{\widetilde{#1}}
	\newcommand{\cN}{{\cal N}}

	\newcommand{\mA}{\mathbb{A}}
	\newcommand{\mE}{\mathbb{E}}
	\newcommand{\mG}{\mathbb{G}}
	\newcommand{\mR}{\mathbb{R}}
	\newcommand{\mS}{\mathbb{S}}
	\newcommand{\mU}{\mathbb{U}}
	\newcommand{\mV}{\mathbb{V}}
	\newcommand{\mW}{\mathbb{W}}
	
	\newcommand{\uq}{\underline{q}}
	\newcommand{\ubq}{\underline{\boldsymbol q}}
	
	\newcommand{\red}[1]{\textcolor[rgb]{1,0,0}{#1}}
	\newcommand{\gre}[1]{\textcolor[rgb]{0,1,0}{#1}}
	\newcommand{\blu}[1]{\textcolor[rgb]{0,0,1}{#1}}

	\newcommand{\cd}{\theta_d}
	\newcommand{\cg}{\theta_g}
    \newcommand{\cc}{\theta_c}
    \newcommand{\cu}{\theta_u}
	\newcommand{\cp}{\theta_p}
	\newcommand{\clb}{\Theta_\alpha^L}
	\newcommand{\cub}{\Theta_\alpha^U}
	\newcommand{\ct}{\Phi_\alpha}
	\newcommand{\vp}{\theta_\alpha}
	\newcommand{\ra}{R_\alpha}
	\newcommand{\dm}{d_{\min}}
	
	\title{On the Rate-Distortion-Complexity Tradeoff for Semantic Communication}
	
	\author{Jingxuan~Chai, Yong~Xiao, \IEEEmembership{Senior~Member, IEEE}, Guangming~Shi, \IEEEmembership{Fellow, IEEE}

    \thanks{*This work is accepted at IEEE Internet of Things Journal. Copyright may be transferred without notice, after which this version may no longer be accessible. 
    
    The work of Y. Xiao was supported in part by the National Natural Science Foundation of China (NSFC) under grant 62525109 and the Mobile Information Network National Science and Technology Key Project under grant 2024ZD1300700. The work of G. Shi was supported in part by the NSFC under grants 62293483, 62506281, 62476206, and the National Key R\&D Program of China under grant 2025YFF0514800. 
    
			J. Chai is with the School of Artificial Intelligence at the Xidian University, Shaanxi 710071, China
			(e-mail: chaijingxuan@stu.xidian.edu.cn)
			
			Y. Xiao is with the School of Electronic Information and Communications at the Huazhong University of Science and Technology, Wuhan 430074, China, also with the Pengcheng Laboratory, Shenzhen, Guangdong 518055, China, and also with the Pazhou Laboratory (Huangpu), Guangzhou, Guangdong 510555, China 
            (e-mail: yongxiao@hust.edu.cn).
			
			G. Shi is with the Peng Cheng Laboratory, Shenzhen, Guangdong 518055, China, and with the School of Artificial Intelligence, the Xidian University, Xi’an, Shaanxi 710071, China
            (e-mail: gmshi@xidian.edu.cn).
			

            Copyright (c) 2026 IEEE. Personal use of this material is permitted. However, permission to use this material for any other purposes must be obtained from the IEEE by sending a request to pubs-permissions@ieee.org.
	}
	}

	\maketitle
    
	\begin{abstract}
        Semantic communication has emerged as a novel communication paradigm that focuses on conveying the user's intended meaning rather than the bit-wise transmission of source signals. One of the key challenges is to effectively represent and extract the semantic meaning of any given source signals. While deep learning (DL)-based solutions have shown promising results in extracting implicit semantic information from a wide range of sources, existing work often overlooks the high computational complexity inherent in both model training and inference for the DL-based encoders and decoders. 
        To bridge this gap, this paper proposes a rate–distortion–complexity (RDC) framework which extends the classical rate-distortion theory by incorporating the constraints on system distortion, including both the traditional bit-wise distortion metric and statistical divergence-based semantic distance, and complexity measure, adopted from the theories of minimum description length (MDL) and information bottleneck (IB). 
        We derive the closed-form theoretical results of the minimum achievable rate under given constraints on semantic distance and complexity for both Gaussian and binary semantic sources. Our theoretical results show a fundamental three-way tradeoff among achievable rate, distortion, and model complexity. 
        Extensive experiments on real-world image datasets validate this tradeoff and further demonstrate that our information-theoretic complexity measure effectively correlates with practical computational costs, guiding efficient system design in resource-constrained scenarios.
	\end{abstract}

	\section{Introduction}

Semantic communication is a novel communication paradigm that focuses on conveying the meaning of a message rather than symbol-level transmission and delivery \cite{weaver1953recent}. It has recently attracted significant interest due to its promising potential to address the critical challenges of network efficiency and the massive data traffic expected in next-generation wireless systems, especially 6G and beyond \cite{shi2021semantic}. 

A primary challenge in semantic communication is efficiently extracting and transmitting task-specific semantic information from a source signal. 
In particular, unlike traditional data-focused communication methods that compress and send every bit of data, semantic communication requires both the transmitter and receiver to share a shared understanding of the specific requirements of each individual task requested by the user \cite{chai2023rate,chai2025on}. 
Recent works mostly adopt DL-based representation learning solutions to convert complex signal sources into compact, low-dimensional representations that preserve the core semantic information. 
Previous results have already shown that the DL-based semantic encoders and decoders can extract and model intricate, non-linear semantic dependencies and context-dependent nuances that are intractable for traditional static, rule-based, and analytically predefined coding solutions \cite{xie2021deep,jiang2022wireless}. 

Despite its promising potential, most existing literature in semantic communication often overlooks a critical drawback: the high computational complexity inherent in both model training and inference of the DL-based encoding and decoding processes \cite{xiao2022rate, liu2022indirect, xiao2022imitation,xiao2025AgentNet}. 
This computational burden presents a major obstacle to the practical deployment of semantic communication systems, particularly in resource-constrained environments. Currently, there is a lack of a unified theoretical framework that can characterize the impact of this complexity on the performance of communication systems. 

This motivates the work in this paper, where we investigate the complexity of DL-based semantic encoding models and their impact on the efficiency and performance of semantic communication, characterized by both the achievable rate of semantic information delivery and the task-specific semantic fidelity of the recovered signals. We develop a unified theoretical framework that can quantify the fundamental tradeoff among rate, semantic fidelity, and model complexity for semantic communication, from the perspectives of information theory and algorithmic complexity. More specifically, to characterize the correlation between the coding complexity and communication efficiency, we quantify the complexity of semantic coding models using the MDL principle, a commonly adopted metric in algorithmic complexity theory that defines model complexity as the minimal description length required to encode both the model itself and its predictions on observed data \cite{gilad2003information,grunwald2007minimum,blier2018description}. 
Compared to other popular metrics, such as the total number of model parameters, MDL offers a comprehensive characterization of coding complexity and the ability to compress and generalize from the source data. It does not depend on the specific model's structure and running platform, and therefore can serve as a more useful metric that quantifies a model's complexity not just by its structure but by its effectiveness as a compression scheme for the data it is trained on. 
To quantify the impact of the model complexity on task-specific semantic communication performance, we combine the theory of MDL with the IB principle \cite{tishby2000information,tishby2015deep}, an information-theoretic principle that characterizes the optimal representation of a given source signal that preserves the maximum amount of semantic information that is relevant to a specific task. 
Moreover, motivated by the recent results suggesting that, compared to the traditional symbol-based distortion measure, the divergence measures of probability distributions are a more suitable metric for evaluating the semantic difference in semantic communication systems\cite{chai2025on}, we consider both the symbol-based distortion and the divergence of probability distributions between the user's original semantic information and the received signals, which we refer to as the semantic distance. 
We derive the closed-form expressions for the achievable rate under the constraints on mean squared error (MSE) distortion, semantic distance, and coding complexity when the source signals follow a Gaussian or binary distribution, which reveals the tradeoff among achievable rate, semantic distance, and model complexity. 
Beyond the theoretical analysis, our experimental study establishes a concrete link between this information-theoretic complexity measure and practical computational costs, demonstrating its direct utility for guiding encoder design in resource-constrained systems. 

The main contributions of this paper are as follows:

    \begin{itemize}
        \item We propose a novel RDC framework that generalizes rate-distortion theory by incorporating a model complexity constraint, explicitly modeling the cost of DL-based coding in semantic communication systems.
        \item We derive the closed-form expressions of the RDC functions for both Gaussian and binary semantic sources,
        which both reveal the tradeoff among rate, distortion, and complexity, namely the RDC tradeoff.
        \item We develop a variational method to approximate the solution to the RDC optimization problem when the data source is unknown and only data samples are available.
        \item Extensive experimental results on classification and generation tasks validate the RDC tradeoff and empirically link the proposed complexity measure to practical computational costs, providing guidance for designing semantic encoders in resource-constrained scenarios.
    \end{itemize}

    \section{Background and Preliminaries}
    \label{sec:background}

    \subsection{Quantifying the Complexity of DL-based Models}

Currently, quantifying the complexity of DL-based models remains a formidable challenge. While several established metrics exist, such as the VC dimension\cite{vapnik1994measuring}, model size, and floating-point operations (FLOPs), each carries specific limitations. For instance, the VC dimension \cite{vapnik1994measuring}, a classical measure derived from statistical learning theory, characterizes a model's expressive power by the size of the largest dataset that can perfectly fit all possible labels. It provides a worst-case bound on generalization error, which is too loose for high-dimensional neural networks. The number of degrees of freedom \cite{gao2016degrees} offers another perspective, measuring the expected gap between training and test error. It quantifies the sensitivity of a model's predictions to perturbations in the training labels, effectively capturing the model's flexibility or adaptability in fitting data. Due to structural and implicit regularization in deep networks, this measure often fails to account for the actual structural richness or the generalization capacity of the learned representations \cite{li2018measuring}. 
    Alternatively, the intrinsic dimension describes the minimum dimensionality at which a model can be effectively trained, typically by optimization within a randomly sampled low-dimensional subspace. This metric reflects the inherent complexity of the learning problem itself rather than the capacity of a specific model architecture.
    From a systems perspective, FLOPs are widely used to measure the computational workload of an algorithm or model. They count the number of floating-point multiplication and addition operations required for a forward pass, providing a hardware-agnostic estimate of runtime and energy consumption \cite{he2016deep,dai2018compressing}. 
    FLOPs serve as a direct and practical proxy for computational complexity in practical deployment scenarios.

    In contrast, information theory, specifically the MDL principle, provides a more comprehensive and robust framework for quantifying the complexity of DL-based models. The MDL-based approach treats model complexity as a communication problem, where the ``best" model is the one that achieves the shortest total compression of both the model parameters and the data given those parameters. Unlike FLOPs or model sizes, MDL bridges the gap between model size and data fitting by measuring complexity in bits. Also, it naturally penalizes overfitting by accounting for the precision of weights: a model with many parameters that require very little precision to describe is viewed as ``simpler" than a smaller model requiring high-precision weights. 

    \subsection{The MDL and IB Principles}

        To quantify model complexity within the RDC framework, in this paper, we adopt the MDL principle to quantify the computational cost of learning beyond simple parameter counting\cite{gilad2003information, grunwald2007minimum, blier2018description}. MDL has recently emerged as a critical metric for evaluating the structural and functional complexity of neural network-based models. Recent advancements have successfully applied MDL frameworks to analyze the structural complexity of Transformer architectures\cite{sefidgaran2023minimum} and to demonstrate that shorter description lengths in large language models correlate with superior generalization\cite{zhao2023large}. Furthermore, variational objectives that approximate MDL-based complexity have been shown to guide models toward asymptotically optimal compression \cite{shaw2025bridging}. 
        Recent information-theoretic frameworks have utilized MDL to account for the finite computational costs of encoding and decoding schemes, a limitation in classical Shannon theory that our RDC framework seeks to address \cite{finzi2026entropy}. By grounding our complexity measure in these recent advancements, we establish a principled link between the mutual information $I(U;X)$ and the effective description length of DNN-based representations. 
        
    Formally, consider a coding model with input $X$ and output $U$, the MDL principle is defined as follows:  
    \begin{definition}
    The MDL principle selects the model \(p(y|x)\) that minimizes the expected code-length, formulated as:
    \begin{equation}
        \min_{p \in \mathcal{P}} \mathbb {E}_q 
        [L(U|X)],
    \end{equation}
    where $q(x,y)$ is the true data distribution,
    $\mathcal P$ is the set of all possible $p(u|x)$,
    and $L(u|x):=-\log_2 p(u|x)$ is the codelength for any $u$ given $x$.
    \end{definition}
    
    Then, for any coding method, the corresponding code-length is lower-bounded by \cite{mackay2003information}:
    \begin{equation}
    \label{eq:pre_mdl}
        \mathbb {E}_q[L(U|X)]\geq H(U|X),
    \end{equation}
    with equality if and only if the coding distribution matches the ground-truth conditional distribution.
    
    By comparing against a classic code scheme \cite{blier2018description} without requiring any model training, which has a code length of $H(U)$, we can then define the gain of the model compression and derive an upper bound as follows: 
    \begin{equation}
        H(U)-\mathbb {E}_q[L(U|X)]\leq H(U)-H(U|X)=I(X;U),
    \end{equation}
    where the compression capability, or model complexity, can be quantified by the mutual information between the model's input and output $I(X;U)$.
    Unlike traditional complexity metrics,
    mutual information directly quantifies the amount of meaningful structure captured from the data. 
    Specifically, in the case of overfitting to random labels where the input and output are independent, the mutual information  
    $I(X;U)$ is zero. Thus, $I(X;U)$ serves as a principled measure of effective model complexity \cite{gilad2003information}.

    Applying mutual information to measure the model complexity of the MDL principle coincides with that of the IB principle.
    Specifically, the IB principle formulates the problem of extracting information about the source $Y$ through a correlated observation $X$ \cite{tishby2000information}.
    Denote by $U=f(X)$ the random variable of extracted information, where $f(\cdot)$ is a function of $X$.
    Formally, given $p_{YX}$, the IB problem is defined as
    \begin{equation}
    \label{eq:IB_problem}
        \min_{p_{XU}}I(X;U),
        \quad \text{s.t.}\;
         I(U;Y)\geq \beta,
    \end{equation}
    Here $I(U;Y)$ represents the  relevance of $U$ to the source $Y$,
    and $I(X;U)$ quantifies the complexity of $U$,  where complexity here is measured by the description length at which the input $X$ is compressed.
    The IB method can be viewed as a rate-distortion problem with the constraint on KL divergence, which is equivalent to the relevance term \cite{tishby2015deep}.

    \subsection{Semantic Rate-Distortion Theory}
    The semantic communication paradigm fundamentally shifts the design objective from the accurate recovery of a source signal to the delivery of intended meaning to accomplish specific tasks \cite{shi2021semantic}. This shift has led to the development of the semantic rate-distortion theory, which reformulates the classical rate-distortion problem from a task-oriented semantic fidelity perspective \cite{chai2023rate, strinati20216g, kountouris2021semantics}.
    To model the intrinsic and unobservable nature of semantic meaning, most of the semantic rate-distortion frameworks assume the encoder can only observe a corrupted version of the semantic source, where the correlation between the semantic source and the observation is quantified by a given conditional distribution\cite{xiao2022imitation,xiao2023reasoning}.
    Formally, let $S$, $U$, $\hat S$ denote the semantic source, indirect observation and reconstruction, the indirect semantic rate-distortion problem is formulated as
    \begin{equation}
        \min_{p_{\hat S|U}} I(U;\hat S),\quad \text{s.t. } \mathbb{E}[d(S,\hat S)]\leq \cd,
    \end{equation}
    where $\cd$ is the maximum tolerable distortion.
    This formulation is analogous to the classical indirect rate-distortion problem \cite{witsenhausen2003indirect}.
    Recent work establishes an indirect rate-distortion framework for such semantic sources \cite{liu2022indirect}. 
    This framework has been extended to incorporate side information and distinct distortion measures at the encoder and decoder, enabling strategic semantic coding \cite{xiao2022rate}.
    Beyond bit-level fidelity, preserving the statistical properties of the semantic source is crucial for many intelligent tasks.
    The rate-distortion problem under the constraint on divergence between the marginal distributions of the source and its recovery is termed as the rate-distortion-perception (RDP) theory \cite{blau2018perception,blau2019rethinking}, often used to learn effective representation for specific tasks such as generation \cite{zhang2021universal} and classification \cite{liu2019classification}.
    Subsequent works have extended the RDP frameworks into indirect source scenarios for task-oriented semantic communications\cite{chai2023rate,chai2025on,wang2025task}. 
    Formally, the indirect semantic rate-distortion problem with a divergence constraint is defined as:
    \begin{align}
        &\min_{p_{\hat S|U}}I(U;\hat S),\\
        \text{s.t. }  \mathbb{E}[d(S,&\hat S)]\leq \cd,\;
        d_{p}(p_S,p_{\hat S})\leq \cp\nonumber
    \end{align}
    where $d_p$ is any divergence metric  and $\cp$ is the maximum tolerable divergence.
    However, these works often overlook the computational cost of the encoding process, limiting the practical deployment of their coding schemes in resource-constrained systems. 

    Recent advancements of semantic rate-distortion theory have incorporated the IB method, demonstrating its potential in task-oriented communications \cite{shao2021learning,barbarossa2023semantic,xie2023robust}. 
    For instance, the authors of \cite{shao2021learning} propose an IB-based coding scheme in which the encoder generates compressed feature vectors that retain the most relevant information about the source.
    Similarly, the authors of \cite{barbarossa2023semantic}
    introduce a semantic communication system incorporating an IB-based source encoder, further emphasizing the importance of relevance in information compression. 
    However, these studies primarily focus on optimizing the relevance term of the IB framework, often ignoring the critical constraints imposed by model complexity at the user end. 
    The most recent and directly relevant advance is presented in \cite{wang2025task}, which formulates a classification-oriented compression problem under constraints on distortion, divergence, and a classification constraint. In this framework, the distortion and divergence terms are measured between the observed signal and the reconstruction. The classification constraint is defined as the conditional entropy of the source given the reconstruction.
    However, the “classification complexity” constraint is equivalent to the relevance constraint in the IB principle. It thus quantifies the uncertainty of the source conditioned on the reconstruction rather than the encoder's computational complexity. Furthermore, since its distortion and perception objectives aim to recover the indirect observation rather than the original semantic source, the framework may, in principle, not guarantee semantic fidelity to the underlying meaning.

    To address the problems above, we introduce a novel semantic rate-distortion framework that constrains the encoder's model complexity using MDL- and IB-based complexity measures, as well as distortion and divergence constraints that directly compare bit-wise and distribution-wise dissimilarity between the semantic source and its recovery. 
    
    \section{System Model and Problem Formulation}
    \label{sec:system_model}
    \subsection{System Model}
    \begin{figure}[ht]
		\centering
        \includegraphics[scale=0.88]{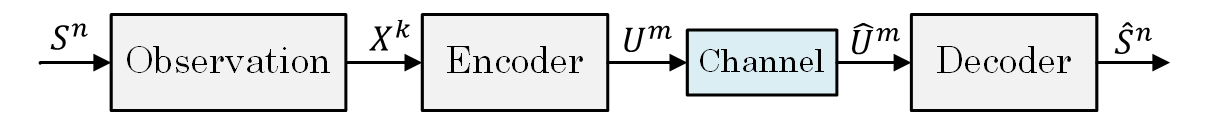}
		\caption{Illustration of the semantic communication model. }
		\label{fig:framework_proposed}
    \end{figure}
    We follow the same line as \cite{xiao2022rate} and consider a semantic communication system where the semantic source $S$ represents the intrinsic meaning or knowledge to be transmitted, as shown in Fig.~\ref{fig:framework_proposed}. This source may involve latent states or contextual information that cannot be directly observed by the transmitter due to hardware or resource limitations. Let $S^n$ denote an $n$-length sequence of independent and identically distributed (i.i.d.) samples drawn from $S$.
    At the transmitter, the encoder obtains $k$ indirect observations $X^k$ about  $S$, where the correlation between the semantic source $S$ and the indirect observation $X$ is characterized by a predetermined conditional probability distribution $p_{X|S}$. These observations are encoded into an $m$-length sequence $U^m$, which serves as a compressed representation of $X^k$. The encoding process is modeled by the conditional distribution $p_{U|X}$.
    The encoded sequence $U^m$ is transmitted over a noisy channel, and the decoder receives a corrupted version, denoted by $\hat U^m$. Using this received sequence, the decoder reconstructs the semantic information as $\hat{S}^n$. The entire system forms a Markov chain $S\to X\to U\to \hat S$.
    For simplicity and tractability of analysis, we focus on the case where the codelengths of the input and output sequences are equal.

    \subsection{Problem Formulation}
    In this paper, we would like to establish a semantic rate-distortion theory for analyzing the fundamental tradeoffs among the minimum achievable rate, system distortion, and model complexity.
    Our objective is to obtain the optimal coding scheme that minimizes the achievable rate under constrained model complexity while achieving a specified distortion level.
    We term this optimization problem the RDC problem.
    Formally, for any given $\cd,\theta_p,\theta_c\geq 0$,
    the RDC problem is formulated as:
    \begin{subequations}
    \label{eq:problem_rdc}
        \begin{align}
    	&\hspace{10pt}\min_{p_{XU\hat S}}I(U;\hat S)\hspace{30pt} \text{(Rate)}\\
        &\label{eq:distortion_term}
        \text{s.t.}\;
         \mathbb{E}[d(S,\hat S)]\leq \theta_d \hspace{15pt}\text{(Distortion)}\\
         &\;\;\;\;\;
        \label{eq:divergence_term}
        d_p(p_{S},p_{\hat S})\leq \theta_p \hspace{15pt} \text{(Semantic distance)}\\
        &\;\;\;\;\;
        \label{eq:complexity_term}
        I(X;U)\leq \theta_c \hspace{25pt} \text{(Complexity),}
        \end{align}
    \end{subequations}
    where the minimization is over all stochastic mappings  $p_{XU\hat S}$ that satisfies the Markov chain $S\to X\to U\to \hat S$,
    and $d_p(p_S,p_{\hat S})$ is any statistical divergence between the marginal distributions of the semantic source and its recovery.
    The solutions to the minimization problem in (\ref{eq:problem_rdc})
    is referred to as the RDC function.
    The optimization problem in (\ref{eq:problem_rdc}) consists of four fundamental components:

    \noindent{\bf Rate term:}
    This term quantifies the minimum achievable communication rate, measured as the mutual information $I(U; \hat{S})$ required to transmit the encoded representation $U$ over the channel to recover $\hat{S}$. 

    \noindent{\bf Distortion term:}
    This term measures the symbol-based distortion between the semantic source and its recovery.
    The distortion function $d$ can be chosen as the Hamming distance for discrete sources or the MSE for continuous sources.
    By retaining the conventional bit-level distortion metric, our RDC framework provides a direct generalization of the classical rate-distortion theory.
    
    \noindent{\bf Semantic distance term:}
    This term measures the degree of the task-specific semantic preservation between the semantic source and its recovery.
    In general, different tasks may require distinct divergence metrics and exhibit varying degrees of sensitivity to statistical differences.
    Therefore, the exact realization of the divergence metric $d_p$ depends on the specific task requirements.
    For example,
    in generation tasks, 
    Wasserstein distance serves as an effective divergence metric to generate image samples with high perceptual quality \cite{arjovsky2017wasserstein,blau2019rethinking}.
    One may therefore define the divergence metric as the Wasserstein distance between the marginal distributions of the semantic source and its recovery, written as
    $d_p(p_{S},p_{\hat S}):=d_W(p_S,p_{\hat S}),$
    where $d_W$ is the squared Wasserstein distance, defined as
    \begin{equation}
    d_{W}(p_S,p_{\hat S}) = \inf_{p_{S\hat S}\in \mathcal P_{S\hat S}} \mathbb {E}_{p_{S\hat S}}[\Vert S-\hat S \Vert^2_2],
    \end{equation}
    where $\mathcal P_{S\hat S}$ is the set of joint distribution of $S$ and $\hat S$.
    In classification, KL divergence is recognized as the preferred divergence metric in most scenarios \cite{murphy2012machine}.
    Notably, when $d_p$ is chosen as the KL divergence between $p_{S|X}$ and $p_{S|\hat{S}}$, i.e.,
    \begin{equation}
    \label{eq:distortion_IB_temp}
    d_\mathrm{KL}(p_{S|X} \Vert p_{S|\hat{S}}),
    \end{equation}
    it becomes equivalent to the mutual information $I(S; \hat{S})$, which corresponds precisely to the relevance constraint $I(U; S)$ in the IB framework~\cite{tishby2015deep}. 
    Thus, the divergence term in our RDC framework naturally generalizes the IB formulation, while allowing for other divergence measures tailored to different semantic tasks.
    %

    \noindent{\bf Complexity term:} 
    This term quantifies the encoder's model complexity, and $\cc$ is the maximum allowable complexity degree.
    Unlike the IB framework, which seeks a compressed representation $U$ of $X$ that preserves relevant information about $S$ while minimizing complexity,
    our proposed RDC approach focuses on optimizing recovery $\hat S$ to preserve the most relevant semantic information from the source.
    In practical scenarios, the complexity term $I(X;U)$ quantifies the DNN-based encoder's model complexity deployed at the user end.
    This formulation aligns well with emerging semantic communication systems, 
    where efficient encoding and decoding are critical for deployment in resource-constrained environments.

    \section{RDC for Gaussian and Binary Semantic Sources}
    \label{sec:particular}
    In this section, we propose closed-form solutions to the RDC problem for two specific types of semantic sources: Gaussian and binary, which allow us to investigate how each component of the RDC problem interacts with the others.

    \subsection{Gaussian Semantic Source}
    \newcommand{\vo}{\rho}
	\newcommand{\vs}{\eta}
    \newcommand{\dx}{\theta_u}

    To better connect the theoretical results to real-world data distributions,
    we first consider a special case of a continuous semantic source, the Gaussian source.
    We consider the case where the semantic source follows a Gaussian distribution, denoted as $S\in\mathcal N(0,1)$. 
	The correlation between the semantic source $S$ and indirect observation $U$, and that between $X$ and the output representation $U$  are characterized as
    \begin{align}
    \label{eq:gaus_assumption1}
    S=&\gamma X+\sqrt{1-\gamma^2}Z_1,\\
    \label{eq:gaus_assumption2}
        X=&\rho U+\sqrt{1-\rho^2}Z_2,
    \end{align}
	where $\gamma,\rho\in(0,1]$,
	$Z_1\in\mathcal N(0,1)$ is a random variable independent of $X$,
    $Z_2\in\mathcal N(0,1)$ is also a random variable independent of $U$.
    The linear assumption in (\ref{eq:gaus_assumption2}) is well-justified, as it aligns with the theoretical guarantees of the Gaussian Information Bottleneck (GIB) framework that the optimal solution to the GIB problem is indeed a linear projection \cite{chechik2003information}.
    %
    In this Gaussian case, we adopt MSE and Wasserstein distance as the constraints on distortion and semantic distance metrics, respectively, defined as
    \begin{equation}
        \label{eq:gaus_c1}
        \mathbb{E}[d(S,\hat S)]\leq \cd,\quad d_W(p_S,p_{\hat S})\leq \cp.
    \end{equation}
    Then we have the following theorem with respect to the closed-form
    expression of the Gaussian RDC function:

    \begin{theorem}
    \label{th:gaus}
    For any given $\cd\geq 0$, $0\leq \cp\leq 1$ and $\cc\geq 0$, the RDC function for Gaussian semantic source $\mathcal N(0,1)$ under complexity constraint in (\ref{eq:complexity_term}), the distortion and semantic distance constraints specified in (\ref{eq:gaus_c1}) is
        \begin{align}
        &R^{\mathcal G}(\cd,\cp,\cc)=\nonumber\\
            &\begin{cases}
            -\dfrac12\log\left(\dfrac{\cd}{\gamma\rho}-\dfrac{1-\gamma\rho}{\gamma\rho}(1+\sigma^2)\right)
            & \hspace{-2pt}\text{if }  \theta_1\leq \cd<\theta_2\\
            -\dfrac12\log\left(1-\left(\dfrac{
        1+\sigma^2-\cd}{2\gamma\rho\sigma}\right)^2\right)&
        \hspace{-2pt}\text{if } \theta_2\leq\cd<\theta_3\\
            0 & \hspace{-2pt}\text{if }\cd\geq \theta_3,
        \label{eq:gaus_rdc}
        \end{cases}
        \end{align}
        where $\rho=\sqrt{1-2^{-2\theta_c}}$, $\theta_1=(1-\gamma\rho)(1+\sigma^2)$, $\theta_2=1+\sigma^2-2\gamma\rho\sigma^2$, $\theta_3=1+\sigma^2$,
        $\sigma=1-\sqrt{\cp}$.
    \end{theorem}
    \begin{IEEEproof}
        See \ref{app:proof_gaus}.
    \end{IEEEproof}

    \begin{figure} [t]
		\centering
		\subfloat[]{
			\includegraphics[scale=0.27]{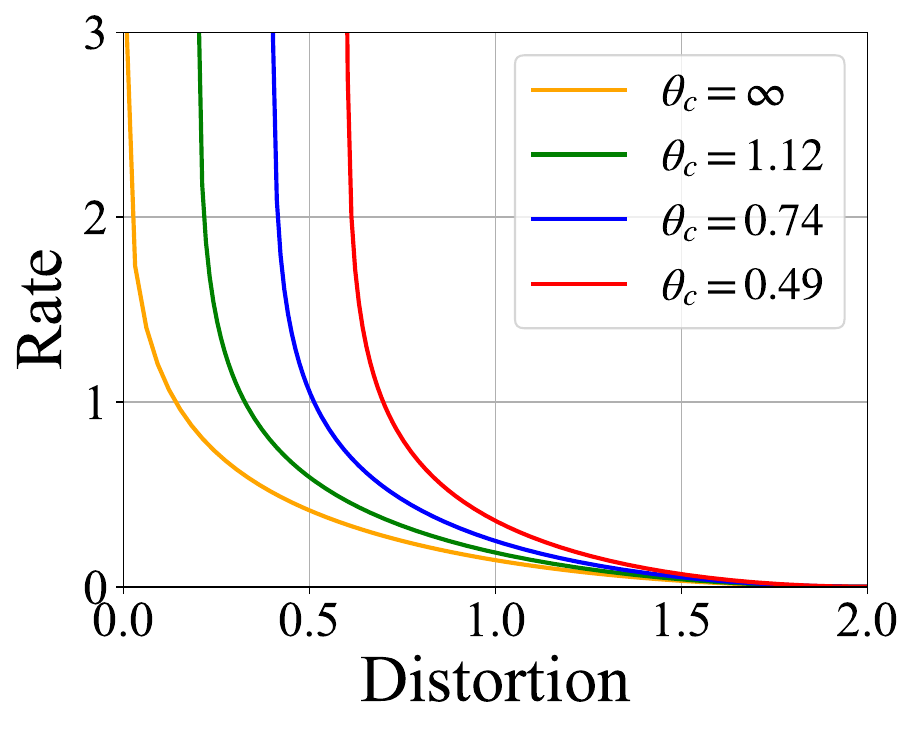}
		}
		\subfloat[]{
			\includegraphics[scale=0.27]{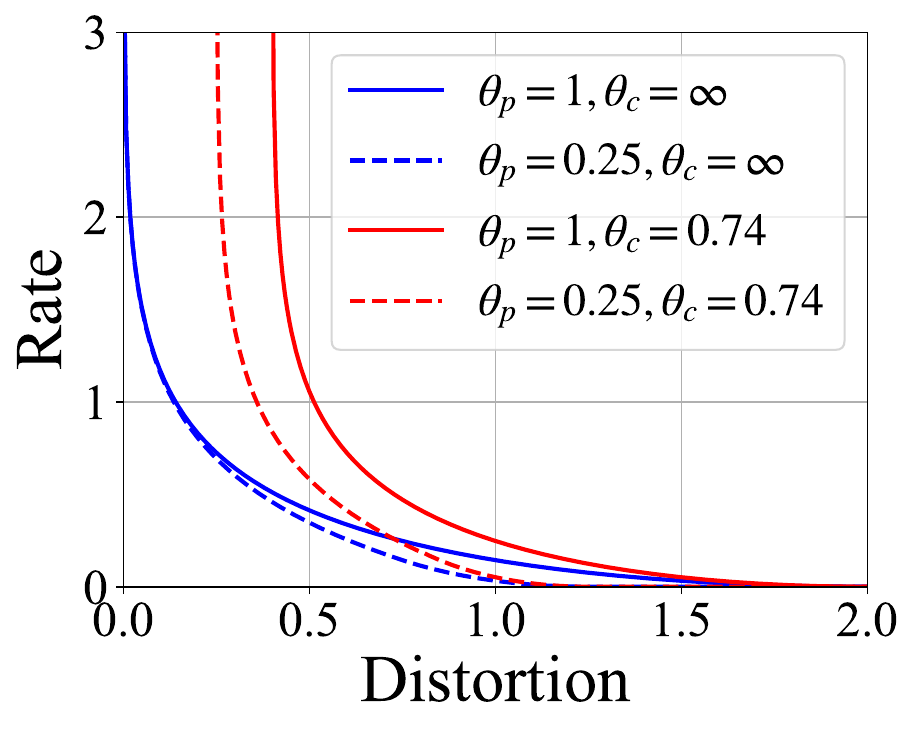} 
        }
		\caption{Curve plots of the RDC functions for Gaussian semantic  sources under
        (a) $\cp=0$ and various $\cc$;
        (b) various $\cp$ and $\cc$.}
		\label{figs:gaus_1} 
	\end{figure}

    In Fig.~\ref{figs:gaus_1}(a),
    we illustrate the Gaussian RDC functions under different values of complexity constraint $\cc$.
    We observe that, for a fixed distortion level, increasing the complexity constraint reduces the achievable rate.
    This corroborates the fact that, even under constrained communication resources, it is possible to achieve a satisfactory distortion level by leveraging additional model complexity at the user end.
    This observation justifies the importance of jointly optimizing communication and computation in semantic communication systems, as model complexity can effectively compensate for limitations in transmission capacity. Such a tradeoff is particularly relevant in practical scenarios where communication bandwidth is limited, but computational power is increasingly abundant at edge devices.
    
    The above results reveal a three-way tradeoff among the minimum achievable rate, distortion, and complexity, namely the RDC tradeoff. 
    Furthermore, Fig.~\ref{figs:gaus_1}(a) illustrates that increasing the complexity level $I(X;U)$  leads to a reduction in the asymptotic minimum distortion.
    This observation reflects the well-known tradeoff in representation learning that higher model complexity, or deeper NN architectures, can achieve better optimal performance. 
    In this context, a higher value of $\rho$, corresponding to greater model complexity, indicates a more precise observation of the model input. 
    We also note that, by setting $\theta_c=\infty$ in Theorem \ref{th:gaus}, the decoder has the observation of the semantic source, i.e., $S=X$.
    This special case reveals a fundamental characteristic of continuous semantic sources that achieving asymptotically perfect recovery, theoretically, requires an encoder with infinite model complexity.

    \begin{figure} [t]
		\centering
        \vspace{-1em}
		\subfloat[]{
			\includegraphics[scale=0.36]{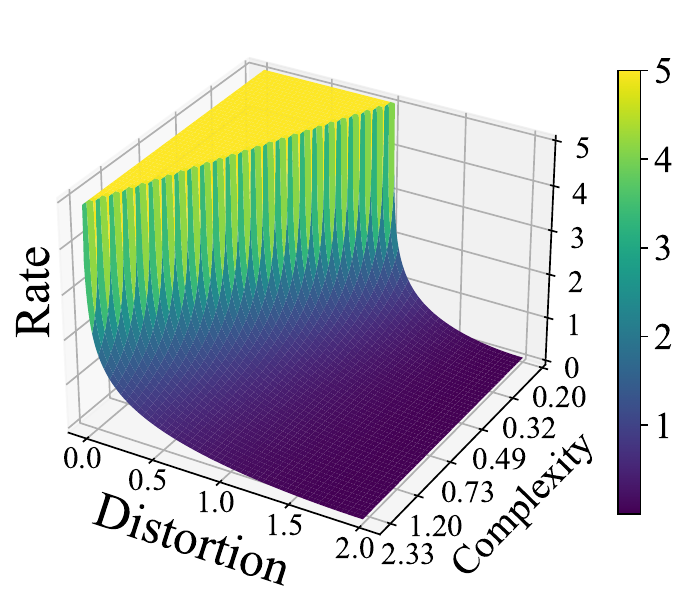}
		}
		\subfloat[]{
			\includegraphics[scale=0.36]{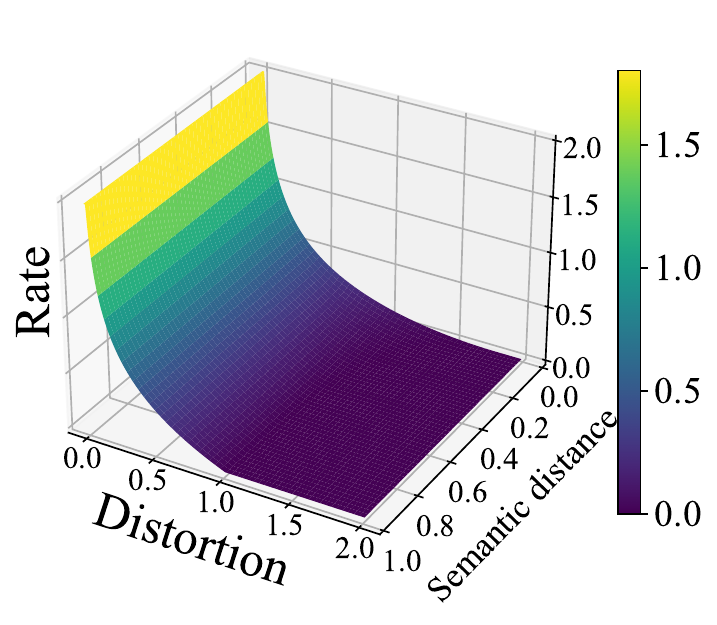} 
        }
		\caption{3D Surf plot of the RDC tradeoff of Gaussian semantic  sources:
        (a) RDC functions under $\cp=0$;
        (b) RDC functions under $\cc=0.74$.}
		\label{figs:gaus_3d} 
	\end{figure}
    We also demonstrate the Gaussian RDC curves under some different degrees of semantic distance, as measured by $\cp$, in Fig.~\ref{figs:gaus_1}(b).
    We can clearly observe that under a given achievable rate, increasing the bit-wise (lower $\cd$) quality leads to lower perceptual quality (higher $\cp$).
    Moreover, an increase in perceptual quality leads to a higher achievable rate.
    The above RDP tradeoff is aligned with the existing works \cite{blau2018perception,wagner2022rate,chai2023rate}.
    
    To obtain a more intuitive understanding of the interaction among the components of the Gaussian RDC problem, in Fig.~\ref{figs:gaus_3d}, we demonstrate the 3D contour plots of Gaussian RDC functions.
    Figs.~\ref{figs:gaus_3d}(a) and \ref{figs:gaus_3d}(b) illustrate the RDC tradeoff and the RDP tradeoff respectively \cite{blau2019rethinking}.
    These three-way tradeoff relationships together constitute a four-way tradeoff among rate, distortion, semantic distance, and complexity.
    This means our RDC framework expands a new dimensionality of the traditional lossy source coding problem.

    \begin{figure} [t]
		\centering
		\subfloat[$\theta_c=\infty$]{
			\includegraphics[scale=0.27]{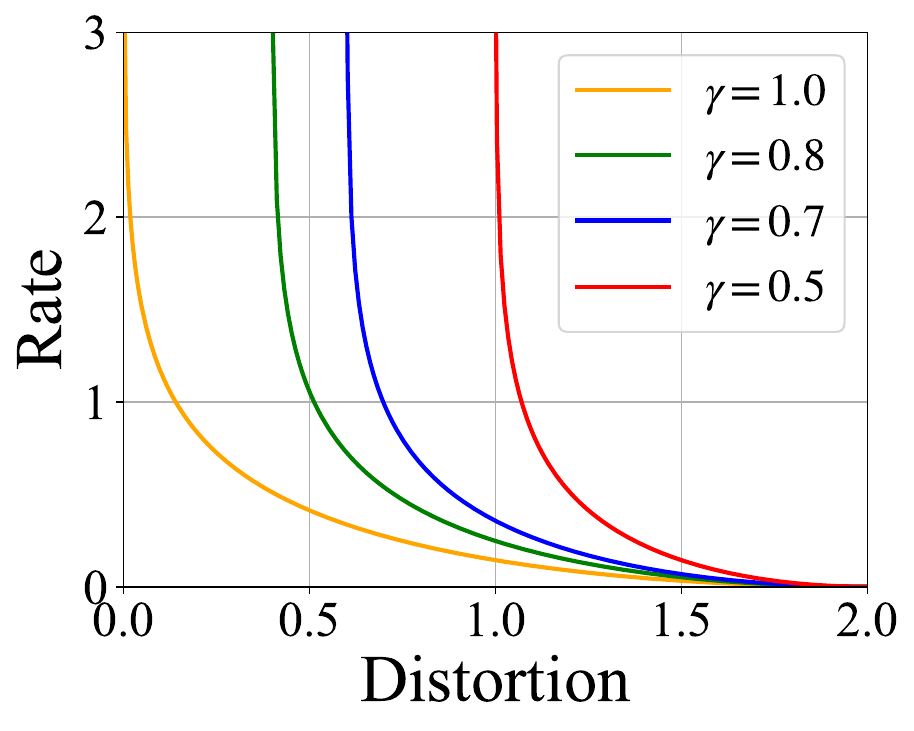}
		}
		\subfloat[$\theta_c=1.12$]{
			\includegraphics[scale=0.27]{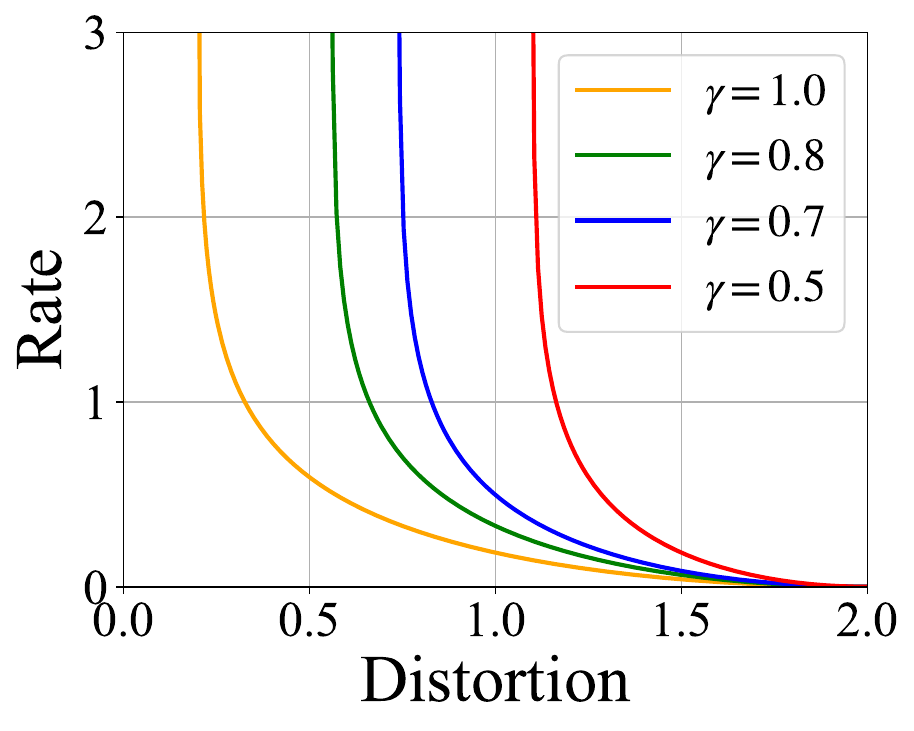} 
        }
		\caption{Curve plots of the RDC functions for Gaussian semantic  sources under $\cp=0$ and
        (a) $\cc=\infty$;
        (b) $\cc=1.12$.}
		\label{figs:gaus_ind} 
	\end{figure}

    In Fig.~\ref{figs:gaus_ind}, we illustrate the Gaussian RDC curves with different $\gamma$ values,
    to observe how the process of indirect observation impacts the behavior of RDC functions.
    From Fig.~\ref{figs:gaus_ind}(a), we observe that for fixed distortion level, when $\gamma$ decreases, the achievable rate increases.
    This can be interpreted that the dependency between the source $S$ and the observation $X$  decreases as $\gamma$,
    which means the encoder has a more disturbed observation of the source.
    In this case, the encoder requires higher rates to achieve the same distortion level as the direct observation case.
    
    In Fig.~\ref{figs:gaus_ind}(b), we observe that,
    even under the case when the encoder has a direct observation of the source ($\gamma=1$, or equivalently, $X=S$),
    the minimum asymptotic distortion $\theta_1$ is greater than zero, since $\theta_1$ is determined by both $\gamma$ and $\theta_c$.
    This also indicates that the impact of the complexity constraint is similar to that of the indirect observation.
    Moreover, under the complexity constraint of finite $\theta_c$, the output representation $U$ can be viewed as an indirect observation of the encoder input $X$, with the uncertainty induced by the conditional probability $p_{U|X}$ of the encoder.
    Therefore, only when the encoder has a direct observation $X=S$ with infinite complexity ($\theta_c=\infty$, or equivalently, $U=X$), the minimum asymptotic distortion equals zero ($\theta_1=0$).
    , as shown in the RDC curve when $\gamma=1,\cc=\infty$ of Fig.~\ref{figs:gaus_ind}.
    
    To have a more intuitive insight into this phenomenon, we consider a special case of Gaussian RDC.
    Specifically,
    taking $\cp=1$ in \eqref{eq:gaus_rdc},
    we have
    \begin{equation}
         R^{\mathcal G}(\cd,1,\cc)=
         \frac12\log\frac{\gamma\rho}{\cd+\gamma\rho-1}
    \end{equation}
    for $1-\gamma\rho\leq\cd\leq 1$.
    This is equivalent to the indirect rate-distortion function with the covariance between the source and the observation being $\gamma\rho$.
    This means the parameters of indirect observation $\gamma$ and the model complexity $\theta_c$ together govern the degree of inevitable uncertainty of the system.

    The result of RDC function in Theorem \ref{th:gaus} generalizes the well-known RDP theory.
    Formally, taking $\cc=\infty$ and $\gamma=1$ in \eqref{eq:gaus_rdc}, which means that the complexity constraint is inactive and the encoder has a direct observation of the source $S$, we have $\rho=1$ and thus
    \begin{equation}
    \hspace{-8pt}
    \label{eq:gaus_rdp_traditional}
        R^{\mathcal G}(\cd,\cp,\infty)
        \hspace{-1pt}
        =\hspace{-1pt}
            \begin{cases}
            \frac12\log\frac1{\cd}
            &\hspace{-5pt} \text{if }  0\leq  \cd<\theta_2'\\
            \frac12\log\frac{1}{1-\left(\frac{1+\sigma^2-\cd}{2\sigma}\right)^2}& \hspace{-5pt}
        \text{if } \theta_2'\leq\cd<\theta_3\\
            0 &\hspace{-5pt} \text{if }\cd\geq \theta_3.
        \end{cases}
        \hspace{-6pt}
    \end{equation}
    where $\theta_2'=1-\sigma^2$ and $\theta_3$ are as defined in Theorem \ref{th:gaus}.
    This result is aligned with the existing results of the Gaussian RDP function under the Wasserstein distance (see Theorem 1 of \cite{zhang2021universal} and Theorem 1 of \cite{chai2025on}).

    \subsection{Binary Semantic Source}
    \label{subsec:binary}
    \begin{figure} [t]
		\centering
		\subfloat[]{
			\includegraphics[width=0.48\linewidth]{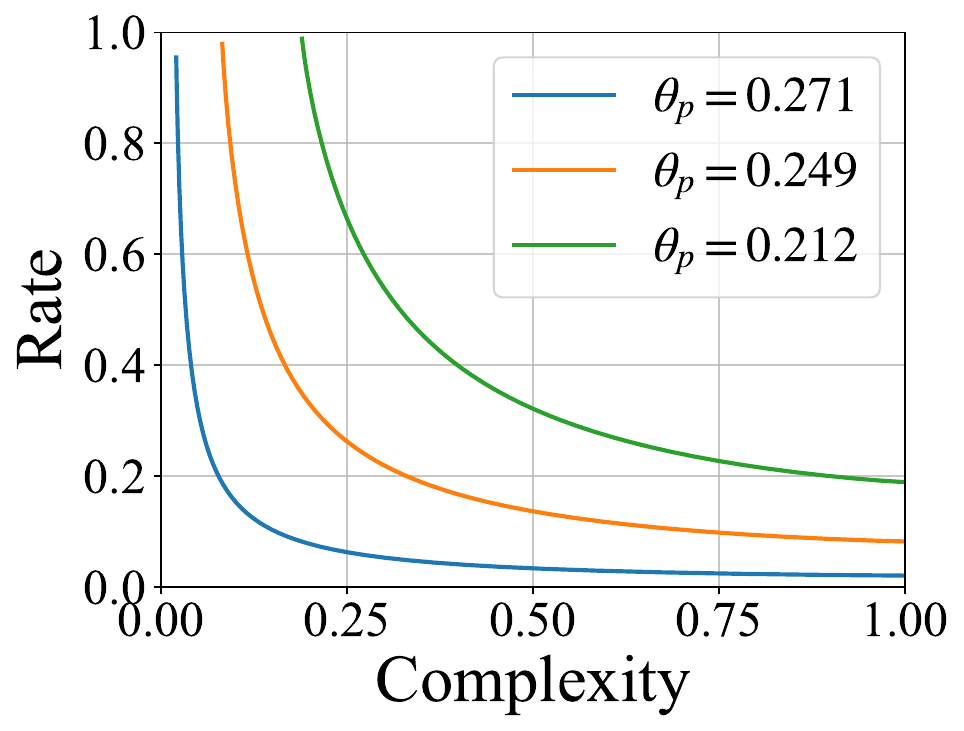}
			\label{fig:bin_1a}	
		}
		\subfloat[]{
			\includegraphics[width=0.46\linewidth]{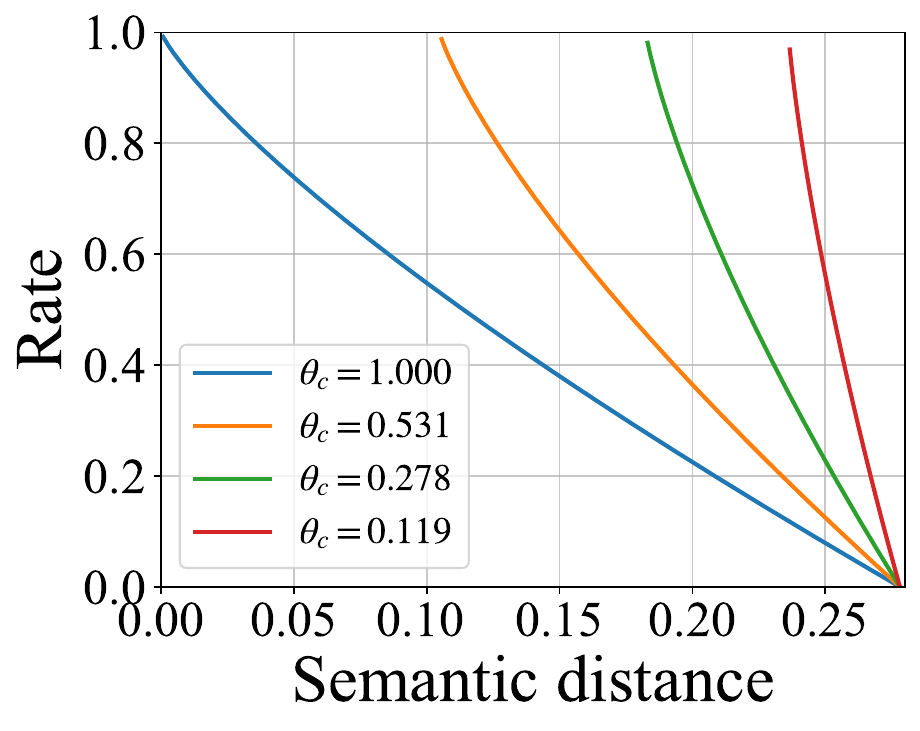} 
			\label{fig:bin_1b}
        }
		\caption{Curve plots of binary RDC tradeoffs:
        (a) Rate-complexity tradeoff;
        (b) Rate-distortion tradeoff.}
		\label{figs:bin} 
	\end{figure}
    
    We also consider another special case of a discrete semantic source,
    where the source follows a Bernoulli distribution.
    Formally, assume the semantic source $S$ and the indirect observation $X$ are doubly symmetric binary sources with a given crossover probability $q_{SX}\in[0,\frac12]$.
    %
    In this case, we adopt the generalization constraint of the IB problem, 
    defined as
    \begin{equation}
    \label{eq:bin_distortion_IB}
        d_\mathrm{KL}(p_{S|X} \Vert p_{S|\hat S})\leq \cp
    \end{equation}
    Then we have the following theorem with respect to the closed-form expression of the binary RDC function:
    \begin{theorem}
    \label{th:bin}
    Under the complexity constraint in (\ref{eq:complexity_term}) and distortion constraint in (\ref{eq:bin_distortion_IB}), the RDC function for binary sources has the following closed-form solution:
        \begin{equation}
            R^{\mathcal B}(\cp,\cc)=1-H_b(q_{U\hat S}),
        \end{equation}
        where $H_b(\cdot)$ is the binary entropy function and $q_{U\hat S}$ is determined by the equation
        \begin{equation}
            \cp=H_b(q_{S\hat S})-H_b(q_{U\hat S}).
        \end{equation}
        Here $q_{S\hat S}$ is 
        \begin{align}
		q_{S\hat S}:=&q_{SX}+q_{XU}+q_{U\hat S}+4q_{SX}q_{XU}q_{U\hat S}\nonumber\\
        &-2(q_{SX}q_{XU}+q_{SX}q_{U\hat S}+q_{XU}q_{U\hat S}),
        \end{align}
        where $q_{XU}$ is determined by $\cc=1-H_b(q_{XU})$.
    \end{theorem}

    \begin{IEEEproof}
        See \ref{app:proof_th1}.
    \end{IEEEproof}

    The results of Theorem \ref{th:bin} generalize the solutions to the classical binary IB problem.
    More specifically, when no compression is applied in the proposed coding scheme, i.e., $\hat S=U$, we obtain 
    \begin{equation}
        I(\hat S;S)=I(U;S)=1-H_b(q_{SU}),
    \end{equation}
    which corresponds to the generalization term of the binary IB problem.
    Moreover, the resulting coding scheme aligns exactly with that of the binary IB problem (see Section 3.1.1 of \cite{zaidi2020information}).
    Our work thus not only provides a novel perspective on  RDC tradeoff, but also establishes a direct connection to the well-established information bottleneck theory.
    
    In Fig.~\ref{figs:bin}, we plot the mutual information terms of the RDC problem of binary sources.
    As depicted in Fig.~\ref{fig:bin_1a}, each curve represents the trade-off between the achievable rate and model complexity for fixed values of $\cd$.
    We note that, to achieve the same distortion level, as complexity increases, the achievable rate decreases, and conversely, lower complexity results in a higher achievable rate.
    This corroborates the fact that, even under constrained communication resources, it is possible to achieve a satisfactory distortion level by leveraging additional computational resources at the user end.
    This observation justifies the importance of jointly optimizing communication and computation in semantic communication systems, as computational resources can effectively compensate for limitations in transmission capacity. Such a trade-off is particularly relevant in practical scenarios where communication bandwidth is limited, but computational power is increasingly abundant at edge devices.
    The above three-way tradeoff among rate, distortion, and complexity is consistent with the results of the Gaussian RDC functions in Fig.~\ref{figs:gaus_1}(a).
    The distinct tradeoff revealed by our binary RDC function highlights the importance of explicitly modeling encoder complexity. For comparison, one recent line of work formulates a constraint on $H(S|\hat{S})$ and treats it as a form of complexity measure ~\cite{wang2025task}. 
    Their analysis shows that this constraint does not create an active tradeoff with perception in the binary case. In contrast, our complexity term $I(X;U)$ directly quantifies the encoder's representational cost via the MDL principle. Our binary RDC results demonstrate a clear three-way tradeoff among rate, semantic fidelity, and model complexity, confirming that $I(X;U)$ captures a fundamental and active resource dimension essential for balancing computation and communication in semantic encoding systems.

    \begin{figure} [t]
		\centering
		\includegraphics[scale=0.50]{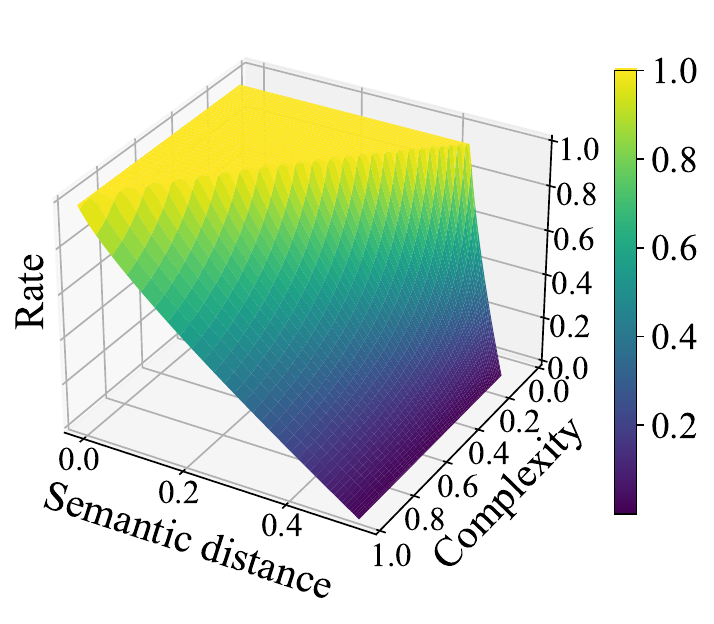}
		\caption{3D Surf plot of the RDC tradeoff of binary semantic sources.}
		\label{fig:bin_3d} 
	\end{figure}

    As illustrated in Fig.~\ref{fig:bin_1b}, we observe that when the complexity satisfies $I(X;U)<1$, achieving zero-distortion recovery of the semantic source becomes theoretically impossible.
    Furthermore, as $I(X;U)$ decreases, the minimum achievable distortion increases.
    This corresponds to the fact that reduced $I(X;U)$ values correspond to higher compression levels in the output representation $U$,
    which inevitably introduces higher uncertainty through the coding process.
    Consequently, when the decoder receives this compressed and uncertain representation, it becomes fundamentally incapable of perfectly reconstructing the original semantic source.
    This observation aligns with the intuitive understanding that excessive compression sacrifices fidelity, thereby increasing the distortion in the reconstructed signal.
    This corresponds to the fact that to achieve a promising distortion level, one should allocate sufficient model complexity at the user end.
    To obtain a more intuitive understanding of this three-way tradeoff, we also present a 3D contour plot of the mutual information terms of the proposed coding problem in Fig.~\ref{fig:bin_3d}.
    The behavior of the binary RDC tradeoff aligns with that of the Gaussian case.
    


    \section{A Variational Approach to RDC Optimization} 
    
    The proposed RDC optimization problem is inherently challenging due to the analytical intractability of the mutual information terms. To address this, we develop a variational framework, called the \emph{Variational Rate-Distortion-Complexity (VRDC)} method, to derive computationally efficient approximations of the optimal solutions. This approach is particularly suited for data-driven scenarios where the underlying source distribution is unknown and must be inferred from empirical samples. 

    Let us first introduce the approach of optimizing the RDC problem for the classification task under the generalization constraint of the IB problem. 
    The optimization problem proposed in (\ref{eq:problem_rdc}) under the generalization constraint in (\ref{eq:distortion_IB_temp}) can be rewritten in the following Lagrange form 
    \begin{equation}
    \label{eq:exp_lb_1}
    	\mathcal L_\text{cls} := I(U;\hat S)+\lambda_c I(X;U)-\lambda_d I(\hat S;S),
    \end{equation}
    where $\lambda_c$ and $\lambda_d$ are Lagrange multipliers. The value of $\lambda_c$ controls the degree of compression, i.e., higher $\lambda_c$ implies higher compression rate;
    And $\lambda_d$ controls the quality of the recovery $\hat S$ of the semantic source $\hat S$: higher $\lambda_d$ implies lower distortion.
    We then establish the variational bound of (\ref{eq:exp_lb_1}) by investigating the bounds associated with each mutual information term in (\ref{eq:exp_lb_1}).
    Firstly, based on \cite{kussul2017deep}, the complexity mutual information term is upper-bounded by
    \begin{equation}
    \label{eq:exp_mi_complexity}
        I(X;U)\leq \mathbb {E}_{p_{XU}}(\log p(u|x)-\log t(u)),
    \end{equation}
    where $t(u)$ is a given distribution
    as a variational approximation of $p(u)$.
    Then for given $q(s|\hat s)$, we have the following bound for the distortion term $I(\hat S;S)$ of (\ref{eq:exp_lb_1}):
    \begin{equation}
    \label{eq:exp_mi_distortion}
        I(\hat S; S)\geq 
         \mathbb {E}_{p_{SXU\hat S}} \log q(s|\hat s)+H(S)
    \end{equation}
    Similarly, can also derive the following bound:
    \begin{equation}
        \label{eq:exp_mi_rate}
        I(U;\hat S)\leq \mathbb {E}_{p_{SXU\hat S}}(\log p(\hat s|u)-\log r(\hat s))
    \end{equation}
    where $r(\hat s)$ is a given distribution as a variational approximation of the distribution recovered signal $p(\hat s)$.
    Combining (\ref{eq:exp_mi_complexity}), (\ref{eq:exp_mi_distortion}) and (\ref{eq:exp_mi_rate}), the loss function in (\ref{eq:exp_lb_1}) is upper-bounded by
    \begin{align}
        \mathcal L_\text{cls}\leq  
        \tilde{\mathcal L}_\text{cls}:= &
        \mathbb {E}_{p_{U\hat S}}\log \frac{p(\hat s|u)}{r(\hat s)}
        +\lambda_c \mathbb {E}_{p_{XU}}\log \frac{p(u|x)}{t(u)}\nonumber\\
        &-\lambda_d \mathbb {E}_{p_{S\hat S}} \log q(s|\hat s).
        \label{eq:loss_cls_final},
    \end{align}
    where the constant entropy term $H(S)$ in \eqref{eq:exp_mi_distortion} is omitted as it does not affect the optimization.
    Hence \eqref{eq:loss_cls_final} is the final loss function for the classification task.

    Similarly,
    when the rate is fixed,
    the loss function for the generation task is formulated as
    \begin{equation}
    \label{eq:loss_gen_original}
        \mathcal{L}_\text{gen}:= \mathbb{E}[d(S,\hat S)] + \lambda_p d_W(p_S,p_{\hat S}) + \lambda_c I(X;U).
    \end{equation}
    Here we use the Wasserstein distance $d_W$ to measure the perceptual quality, and $\lambda_p,\lambda_c$ are the tuning parameters.
    Based on (\ref{eq:exp_mi_complexity}), the loss function in (\ref{eq:loss_gen_original}) is lower bounded by
    \begin{align}
         \tilde{\mathcal L}_\text{gen}\hspace{-2pt}:=\hspace{-2pt}
        \mathbb{E}[d(S,\hat S)]\hspace{-2pt} +\hspace{-2pt} \lambda_p d_W(p_S,p_{\hat S})\hspace{-2pt} + \hspace{-2pt}\lambda_c \mathbb {E}_{p_{XU}}\hspace{-2pt}\log\hspace{-2pt} \frac{p(u|x)}{t(u)}
        \label{eq:loss_gen_lb}
    \end{align}
    where $t(u)$ is a given probability distribution.
    By adopting a WGAN-based architecture in \cite{blau2019rethinking}, the loss function is reformulated as
    \begin{align}
        \tilde{\mathcal L}_\text{gen}=&
         \mathbb{E}[d(S,\hat S)] + \lambda_c \mathbb{E}_{p_{XU}} \log\frac{p(u|x)}{t(u)} \nonumber\\
        &+ \lambda_p \max_{h\in \mathcal H} ( \mathbb{E}[h(S)]-\mathbb{E}[h(\hat S)] )). 
    \end{align}
    By replacing all expectations by sample means,
    the loss function for generation task is approximated by
    \begin{align}
        \tilde{\mathcal L}_\text{gen}\approx&
         \frac1N\sum_{i=1}^N \big(
         [d(s_i,\hat s_i)] + \lambda_c d_\mathrm{KL}(p_{U|X=x_i}\Vert t_U)\nonumber\\
        &+ \lambda_p \max_{h\in \mathcal H} \left( h(s_i)-h(\hat s_i)\right)\big)  
        \label{eq:loss_gen_final},
    \end{align} 
    where $s_i,\hat s_i, x_i$ are the $i$-th sample of the data source, reconstruction and observation respectively.
    Using Monte-Carlo sampling to approximate the above expectations and assuming the variational distributions $t(u), r(u)$ follow the Gaussian distributions, the loss functions in \eqref{eq:loss_cls_final}, \eqref{eq:loss_gen_final}  provide tractable objectives for the associated tasks \cite{kingma2013auto,alemi2017deep}.
    Below, we explicitly demonstrate the derivation of the estimated complexity term in \eqref{eq:exp_mi_complexity}.
    Formally, assume the conditional distribution $p(u|x)$ follows a multivariate Gaussian distribution $p(u|x)=\mathcal N(u|\mu(x),\text{diag}(\sigma (x)))$ for any given $x$,
    where the vector pairs $\mu(x),\sigma(x)$ is
    determined by the output of the DNN-based encoder with input $x$.
    We can use the reparameterization trick to write $p(u|x)du=p(\epsilon)d\epsilon$
    where $\epsilon$ is a Gaussian random variables \cite{kingma2013auto}. Assume that the variational posterior $t(u)$ follows a Gaussian distribution $t_U\sim\mathcal N(0,\mathbf{I})$ \cite{alemi2017deep}.
    Under these assumptions, for given observation $x$, the complexity term \eqref{eq:exp_mi_complexity} can be estimated through:
    \begin{align}
        I(X;U)\leq &\mathbb {E}_{p_{XU}}(\log p(u|x)-\log t(u))=d_\mathrm{KL}(p_{U|X}\Vert t_U)\nonumber\\
        =&\frac12\sum_{i=1}^K\left( \sigma_i(x)^2+\mu_i(x)^2-1-2\log\sigma_i(x)
        \right),
        \label{eq:exp_tractable_complexity}
    \end{align}
    where $K$ is the dimensionality of $U$, $\sigma_i(x)$ and $\mu_i(x)$ are the $i$-th components of $\sigma(x)$ and $\mu(x)$, respectively.
    Thus, the upper bound in \eqref{eq:exp_tractable_complexity} provides a tractable estimation of the mutual information measured complexity $I(X;U)$.

    \section{Experimental Results}
    \label{sec:exp}
    

    \subsection{Experimental Results on Different Tasks}
    We conduct extensive experiments using the VRDC method on two image datasets, to accomplish classification and generation tasks.
    The experimental results are derived by alternating the controllable parameters and optimizing the associated loss functions in 
    \eqref{eq:loss_cls_final} and \eqref{eq:loss_gen_final}  of these tasks respectively.
    The model complexity is obtained by applying stochastic gradient descent (SGD) to the VRDC objective in \eqref{eq:exp_tractable_complexity}

    \noindent\textbf{Classification Task:}
    \begin{figure} [t]
		\centering
		\subfloat[]{
			\includegraphics[width=0.48\linewidth]{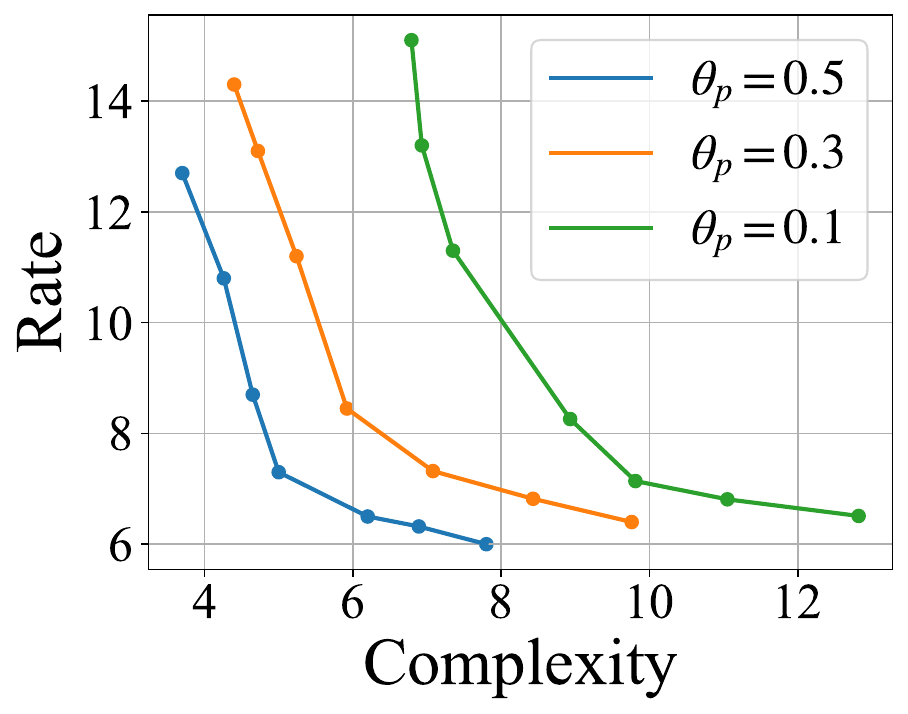}
			\label{fig:exp_1a}	
		}
		\subfloat[]{
			\includegraphics[width=0.46\linewidth]{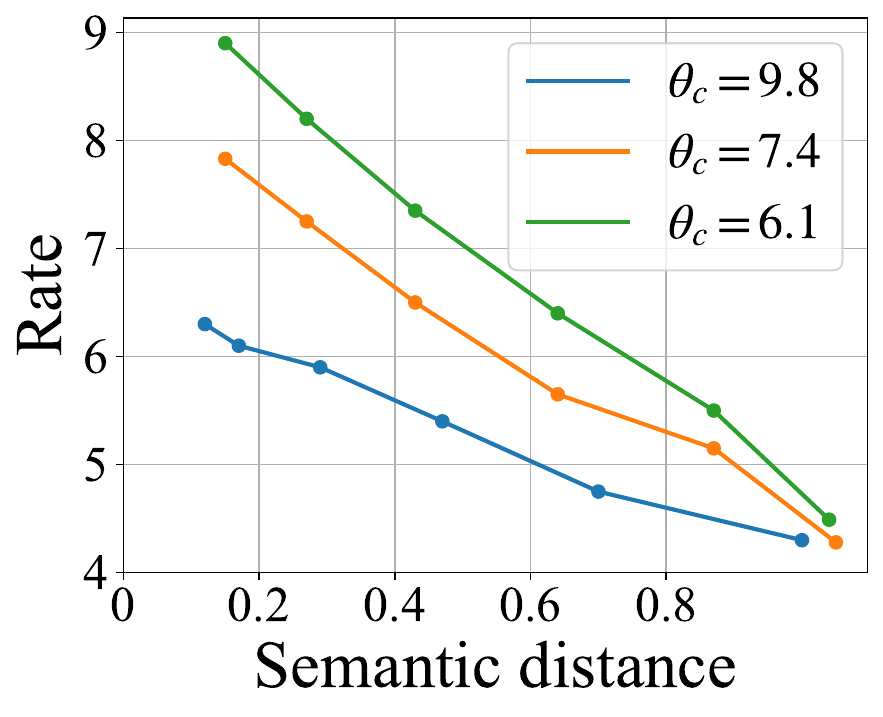} 
			\label{fig:exp_1b}
            }
		\caption{RDC curves of the image classification task:
        (a) Rate-complexity tradeoff curves;
        (b) Rate-semantic distance tradeoff curves.}
		\label{figs:exp_1} 
	\end{figure}
    \begin{figure} [t]
		\centering
		\includegraphics[scale=0.50]{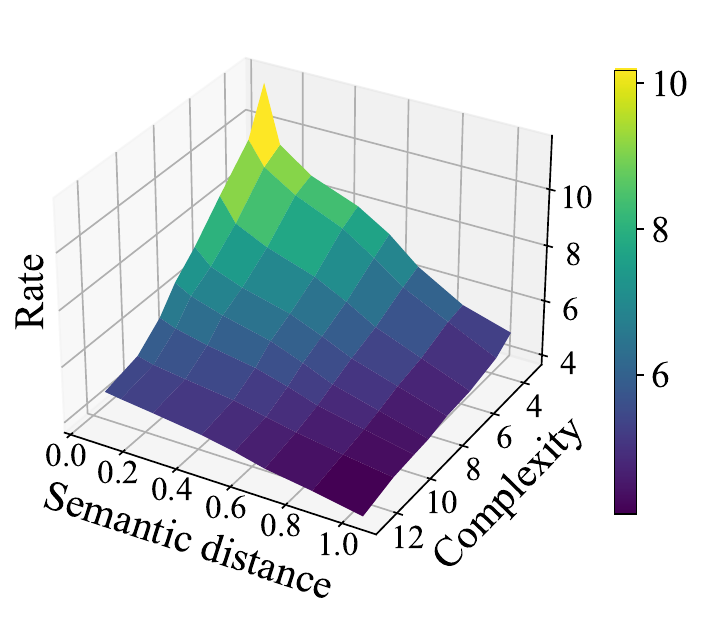}
		\caption{3D surf plot of the simulation RDC tradeoff.}
		\label{fig:exp_3d} 
	\end{figure}
    We demonstrate the classification performance of the VRDC framework, along with that of the traditional DNN-based coding schemes,
    to validate the effectiveness of the RDC methodology.
    We first illustrate the simulative rate-complexity curves under some fixed distortion level and the rate-distortion curves under some fixed complexity degree in Fig.~\ref{fig:exp_1a} and \ref{fig:exp_1b} respectively.
    We can once again observe a three-way tradeoff among the transmission rate, distortion and model complexity for practical image semantic sources.
    Similarly, as demonstrated in Fig.~\ref{fig:exp_1b},
    we can see that under a given complexity degree, an increase in the transmission rate results in lower distortion.
    This again justifies the fact that the user can achieve the same recovery quality with substantially reduced communication resources, at the cost of the increased model complexity at the user end.
    Therefore, to accomplish the intelligent tasks, it is possible to offload the  communication resource at the cost of increasing the model complexity.
    In Fig.~\ref{fig:exp_3d},
    we also present the 3D contour plot showing the RDC tradeoff, which aligns with the theoretical observations derived from the analysis of Gaussian and binary RDC functions.
    
	\begin{figure} [t]
		\centering
		\subfloat[]{
			\includegraphics[width=0.48\linewidth]{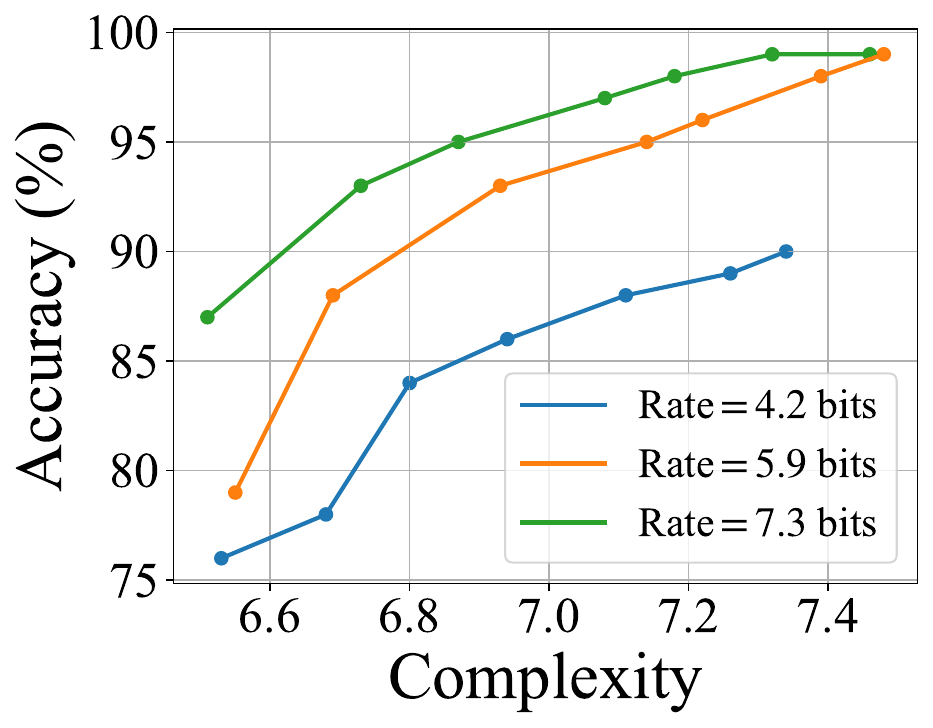}
		}
		\subfloat[]{
			\includegraphics[width=0.48\linewidth]{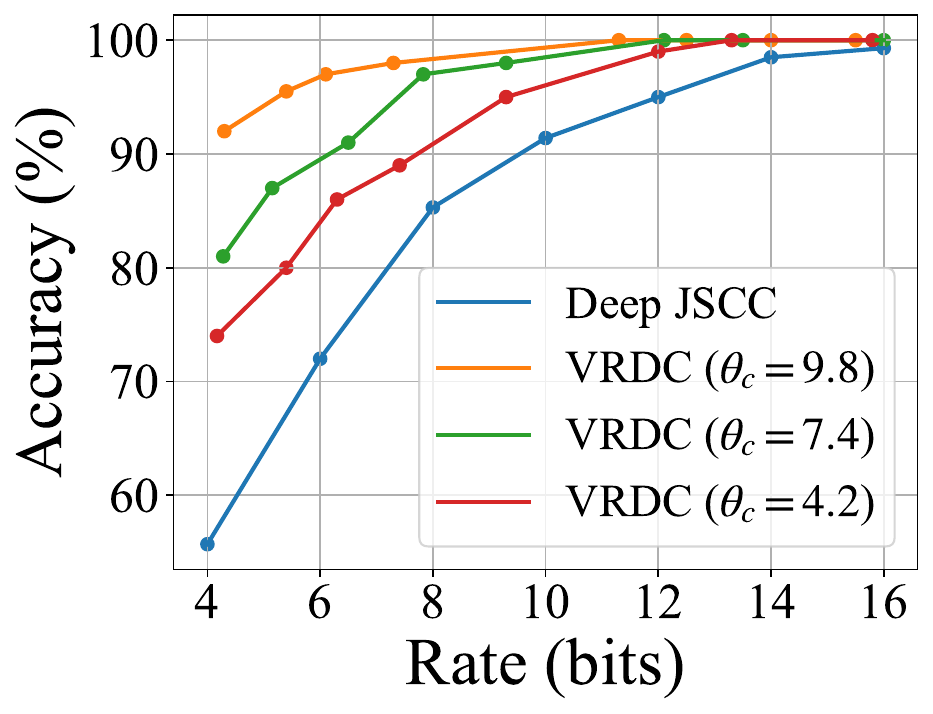} 
            }
		\caption{ (a)
        The classification accuracy of the proposed VRDC method under varying complexity;
        (b)
        The comparison between the classification accuracy of the DeepJSCC \cite{bourtsoulatze2019deep} approach and that of the VRDC method.}
		\label{figs:exp_2} 
	\end{figure}
    
    We also compare the classification accuracy of the proposed VRDC method with the traditional coding scheme that focuses implicitly on the accurate recovery of the image sources.
    More specifically, we consider the DeepJSCC method \cite{bourtsoulatze2019deep} as the benchmark coding scheme of our proposed VRDC method.
    In order to quantify the amount of the transmission rate of DeepJSCC, we adopt a uniform quantizer with $L$ levels to process the output of the DeepJSCC encoder,
    where the soft gradient estimator is used to backpropogate through the quantizer \cite{mentzer2018conditional}.
    
    We compare the classification accuracy of DeepJSCC with the proposed VRDC methods under different transmission rates.
    It can be observed that the proposed method achieves significantly superior performance, particularly in scenarios characterized by low transmission rates and high model complexity of VRDC.
    This is because the traditional DeepJSCC method only focuses on accurately delivering the original image signals, inadvertently transmitting information that is irrelevant to the classification task.
    While the proposed VRDC method optimizes the task-relevant distortion, as measured by KL divergence, thereby generating a more compressed and informative representation for the classification task.

    \begin{figure} [t]
		\centering
		\subfloat[]{
			\includegraphics[width=0.48\linewidth]{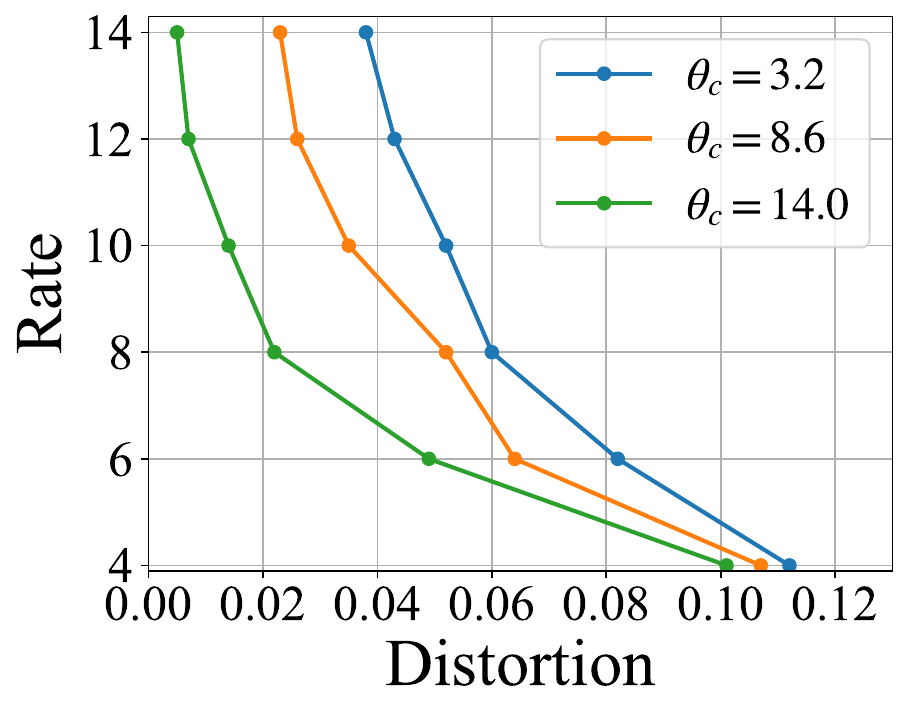}
		}
		\subfloat[]{
			\includegraphics[width=0.48\linewidth]{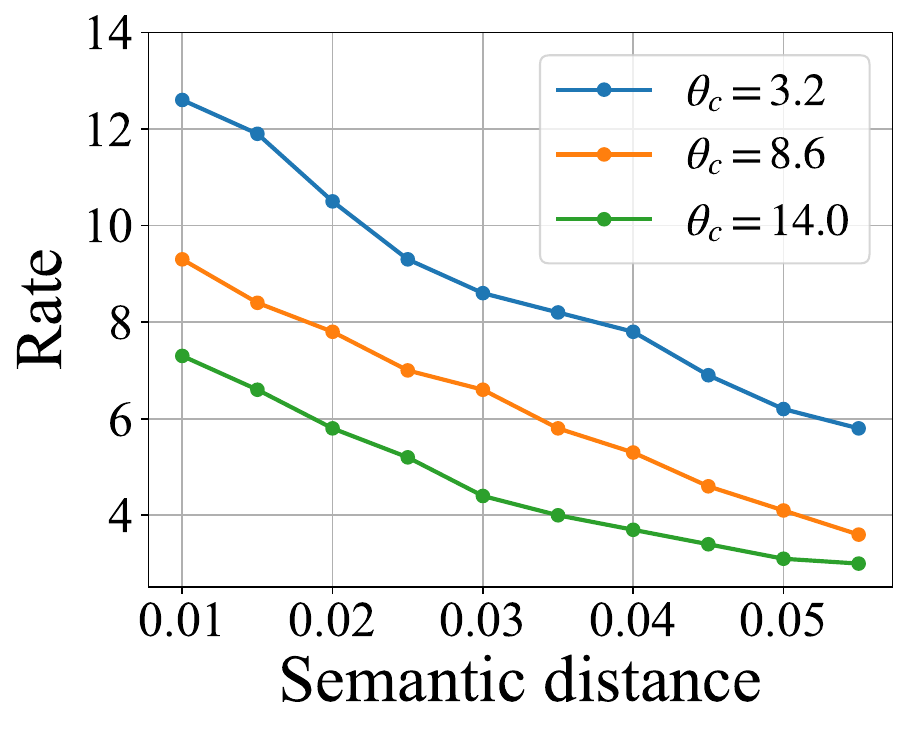} 
            }
		\caption{
        Simulation RDC curves for image generation task under some fixed complexity:
        (a) Rate-distortion tradeoff; 
        (b) Rate-perception tradeoff.}
		\label{figs:exp_gen} 
	\end{figure}

    \noindent\textbf{Image Generation Task}:
    We then consider the image generation task on MNIST dataset, using the VRDC approach
    through optimizing the loss function in (\ref{eq:loss_gen_final}).
    In Figs.~\ref{figs:exp_gen}(a) and \ref{figs:exp_gen}(b),
    we illustrate the curve plots of the experimental RDC functions under various model complexity values.
    We  observe from Fig.~\ref{figs:exp_gen}(a) that, under the same distortion level, increasing the model complexity leads to a decrease in required transmission rate.
    Similarly, Fig.~\ref{figs:exp_gen}(b) also reveals the fact that for a fixed rate, an increase in model complexity can also increase the perceptual quality.
    This again justifies the rationale behind the RDC tradeoff that one can reduce the required transmission rate by increasing the complexity of the DNN-based encoders at the user end, without either additional communication resources or sacrificing the communication fidelity.
    By incorporating the well-known RDP tradeoff \cite{blau2019rethinking},
    the proposed RDC tradeoff is in fact a four-way tradeoff among transmission rate, bit-wise distortion, distribution-wise perception and model complexity.
    This extension represents a significant advancement in semantic rate-distortion theory by formally incorporating model complexity as a fundamental dimension of the optimization space. The RDC framework thus provides a more comprehensive characterization of modern communication systems, where model complexity play a crucial role in balancing rate, distortion, and perceptual quality constraints.
    
    \begin{figure} [t]
		\centering
		\includegraphics[scale=0.31]{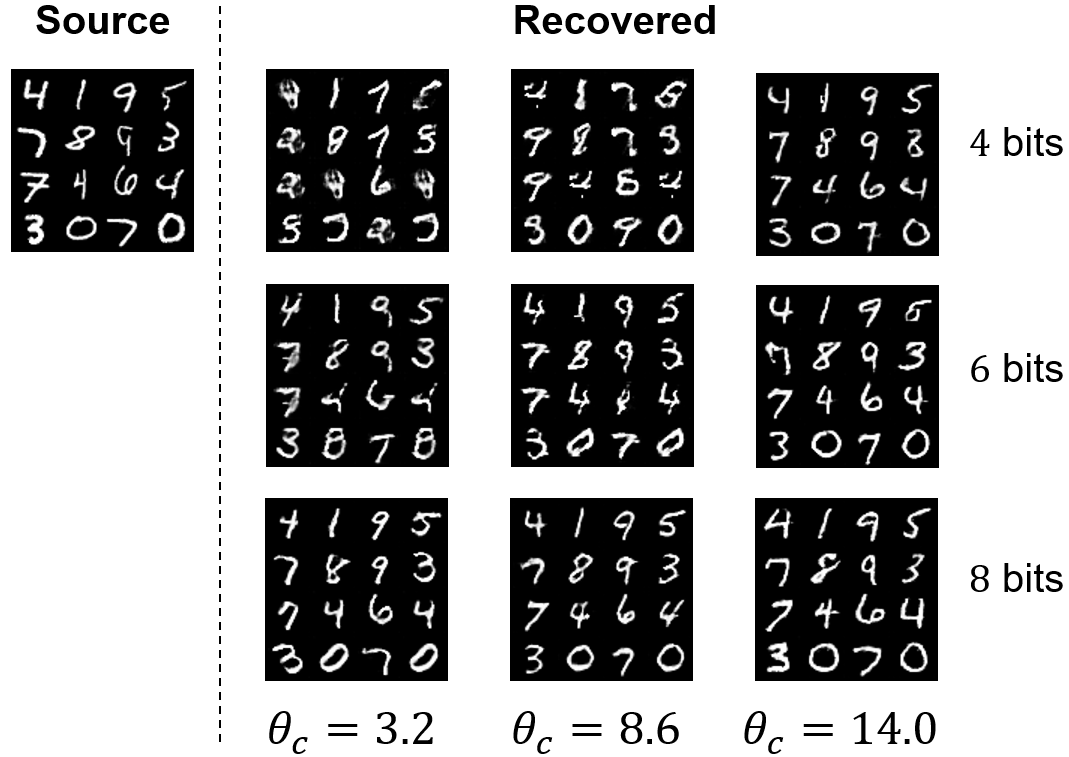}
		\caption{Visual results of the reconstructed image samples across different transmission rate and complexity.}
		\label{fig:exp_MNIST} 
	\end{figure}
    To provide an intuitive demonstration of the RDC tradeoff in generation tasks, Fig.~\ref{fig:exp_MNIST} presents comparative visual results under varying complexity levels. Each row in the figure corresponds to a fixed transmission rate, enabling direct visual assessment of how increased model complexity affects generation quality while maintaining identical rate constraints. 
    We can clearly observe that the perceptual quality of the recovered sample is higher than that generated by the models with lower complexity.
    And an increase in transmission rate can also increase the perceptual quality, when providing the same model complexity.
    The above observations again justify the tradeoff relationship among rate, perceptual quality and complexity.

    \subsection{Relationship between Model Complexity and Computational Complexity}

    \begin{table}[ht]
    \centering
    \caption{Results of computational complexity for classification task on MNIST (top) and CIFAR10 (bottom)}
    \label{tab:classification_flops_mutual_info}
    \begin{tabular}{c|c|c|c}
    \hline
    \textbf{Dim} & \textbf{FLOPs (M)} & \textbf{Complexity}  & \textbf{Acc ($\%$)} \\
    \hline
    8  & 0.489 & 3.52 & 85.5 \\
    16 & 0.493 & 5.87 & 92.2 \\
    32 & 0.501 & 7.94 & 95.3 \\
    64 & 0.517 & 9.23 & 99.0 \\
    128& 0.550 & 9.78 & 99.6 \\
    256& 0.616 & 10.05 & 99.8 \\
    512& 0.747 & 10.21 & 99.5 \\
    \hline
    \end{tabular}

    \vspace{1em}
    
    \begin{tabular}{c|c|c|c}
    \hline
    \textbf{Dim} & \textbf{FLOPs (M)} & \textbf{Complexity}  & \textbf{Acc ($\%$)} \\
    \hline
    32   & 3.263 & 4.23 & 84.7 \\
    64   & 3.312 & 6.15 & 88.3 \\
    128  & 3.410 & 7.82 & 91.2 \\
    256  & 3.607 & 8.54 & 92.6 \\
    512  & 4.028 & 8.93 & 93.4 \\
    1024 & 4.787 & 9.14 & 93.7 \\
    2048 & 6.359 & 9.26 & 93.8 \\
    \hline
    \end{tabular}
    \end{table}

     \begin{table}[ht]
    \centering
    \caption{Results of computational complexity for generation task on MNIST (top) and CIFAR10 (bottom)}
    \label{tab:udim-flops-info}
    \begin{tabular}{c|c|c|c}
    \hline
    \textbf{Dim} & \textbf{FLOPs (M)} & \textbf{Complexity} & \textbf{IS} \\
    \hline
    8  & 0.938 & 4.23 & 1.96 \\
    16 & 0.942 & 7.65 & 2.08\\
    32 & 0.950 & 10.82 & 2.17 \\
    64 & 0.967 & 12.37 & 2.24 \\
    128& 0.999 & 13.56 & 2.29 \\
    256& 1.065 & 14.21 & 2.33 \\
    512& 1.196 & 14.73 & 2.35 \\
    \hline
    \end{tabular}

    \vspace{1em}

    \begin{tabular}{c|c|c|c}
    \hline
    \textbf{Dim} & \textbf{FLOPs (M)} & \textbf{Complexity} & \textbf{IS} \\
    \hline
    32   & 10.683 & 18.76 & 1.18 \\
    64   & 10.715 & 27.33 & 1.24 \\
    128  & 10.813 & 36.83 & 1.33 \\
    256  & 11.010 & 43.50 & 1.42 \\
    512  & 11.403 & 47.21 & 1.49 \\
    1024 & 12.190 & 49.48 & 1.54 \\
    2048 & 13.763 & 50.62 & 1.55 \\
    \hline
    \end{tabular}
    \end{table}

    To further validate that our complexity measure $I(X;U)$ effectively captures the computational cost of the encoder and provides practical guidance for system design, we investigate the relationship between $I(X;U)$, computational complexity, and task performance under varying dimensions of the encoder output $U$. In our experiments, the model complexity $I(X;U)$ is computed using the tractable upper bound derived in \eqref{eq:exp_tractable_complexity}.
     The task performance is measured by classification accuracy and Inception Scores (IS) \cite{salimans2016improved} for classification and generation respectively.
     The computational complexity is measured by recording the number of
     the FLOPs of VRDC frameworks for these two tasks.
     We conduct extensive experiments on both MNIST and CIFAR10 datasets.
     The corresponding experimental results for classification and generation tasks are summarized in Table~\ref{tab:classification_flops_mutual_info} and Table~\ref{tab:udim-flops-info}, respectively.

    We can directly observe that, compared to MNIST dataset, the VRDC framework requires higher model complexity and computational complexity to achieve a promising task performance, since CIFAR10's higher visual complexity demands larger representation capacity and more FLOPs to capture meaningful features. Notably, the relationship between FLOPs and $I(X;U)$ is dataset-dependent.
    As shown in Table \ref{tab:udim-flops-info}, for MNIST, 1.196M FLOPs yields a mutual information of 14.73, while for CIFAR10, 10.68M FLOPs yields only 18.7. This disproportionate scaling is expected because CIFAR10 images contain significantly more complex visual structures that require substantially more computation to extract comparable amounts of task-relevant information.
    
    More importantly, we observe that, when the dimension of $U$ is relatively low, both $I(X;U)$ and FLOPs increase nearly linearly with the dimension. This indicates that $I(X;U)$ serves as a reliable proxy for computational complexity in this regime, confirming that our complexity term can effectively guide the design of practical communication systems where computational resources are constrained.
    We also note that, beyond a certain threshold dimension (e.g., $\text{Dim}=128$ for classification on MNIST; $\text{Dim = 512}$ for generation on CIFAR10), further increasing the dimension of $U$ continues to raise FLOPs, yet $I(X;U)$ saturates and grows slowly. This saturation occurs because $U$ has already captured the maximum amount of relevant information about the semantic source $S$ that can be extracted under the given model structure. Beyond this point, increasing the representational capacity adds redundancy without enhancing the informativeness of the representation, implying that additional computational overhead is unnecessary.
    
    Moreover, beyond the same threshold, task performance closely follows the trend of $I(X;U)$ rather than that of FLOPs. This demonstrates a distinct advantage of $I(X;U)$ over raw computational metrics like FLOPs: $I(X;U)$ not only reflects computational cost but also quantifies the effective complexity, namely the amount of information actually utilized for the task. 
    In other words, $I(X;U)$ captures the point of diminishing returns where further increases in model size cease to improve performance, thereby providing a more principled criterion for balancing complexity, communication rate, and task fidelity in semantic communication systems.

    \begin{figure} [t]
		\centering
		\subfloat[]{
			\includegraphics[width=0.48\linewidth]{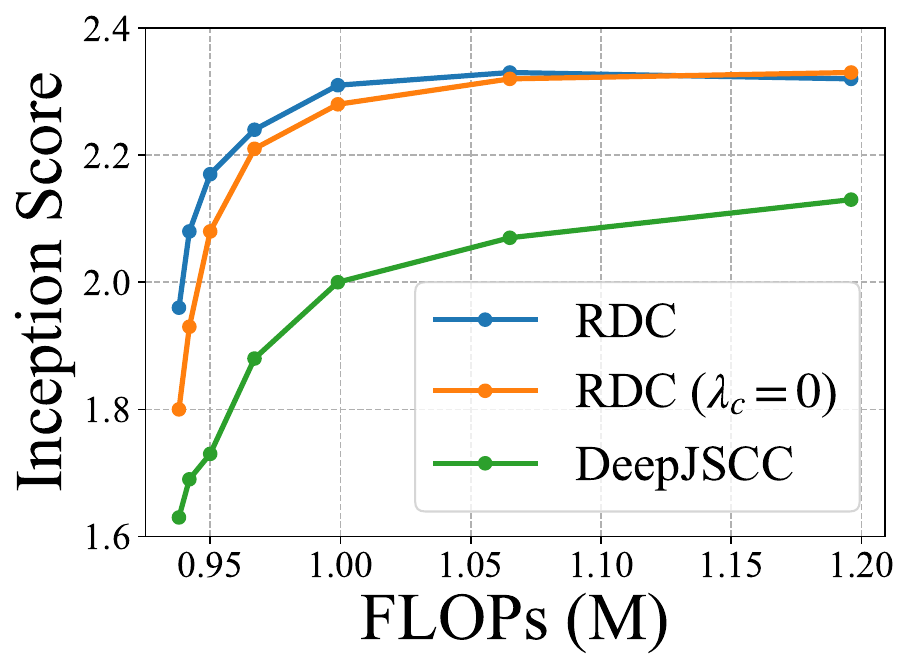}
		}
		\subfloat[]{
			\includegraphics[width=0.48\linewidth]{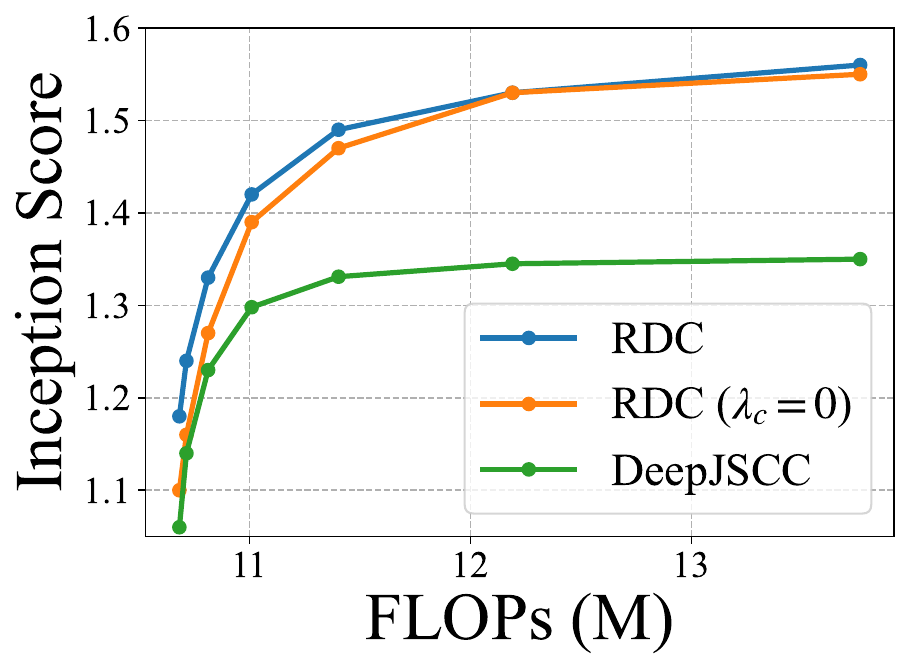} 
            }
		\caption{Curve plots of the relationship between FLOPs and IS across three different models on (a) MNIST and (b) CIFAR10.}
		\label{fig:exp_flops}
	\end{figure}

    To further demonstrate the  benefits of our RDC framework in practical resource-constrained scenarios,
    we also conduct experiments for generation task on both MNIST and CIFAR10 datasets across different models.
    Specifically, we compare the IS of our full RDC model against two baselines, DeepJSCC and  a variant of our framework without the complexity constraint, implemented by setting $\lambda_c = 0$ in the loss function \eqref{eq:loss_gen_final} under the same FLOPs.
    
    As shown in Figs.~\ref{fig:exp_flops} (a) and (b), we note that both RDC and its variant without complexity constraint substantially outperform DeepJSCC across all FLOPs levels on both MNIST and CIFAR10 datasets. 
    This performance gap underscores the importance of incorporating divergence-based semantic distance constraints for generation tasks, as DeepJSCC, which focuses solely on optimizing bit-wise fidelity, fails to preserve perceptual quality effectively. 
    We also observe that, in the low-FLOPs regime (below 0.95M FLOPs on MNIST and 11.01M FLOPs on CIFAR10), the full RDC model achieves noticeably higher IS than the RDP variant. 
    Specifically, on MNIST dataset, when the computational complexity is set at 0.938M and 0.942M FLOPs, the full RDC model achieves roughly 8.89\% and 7.78\% higher generation performance, respectively, compared to the RDC model trained without a complexity constraint.
    On CIFAR10, the improvements of full RDC over RDC without complexity constraint is 7.27\% at 10.683M FLOPs, and 6.92\% at 10.715M FLOPs.
    This advantage stems from the explicit complexity constraint $I(X;U)$, which encourages the encoder to learn a more information-efficient representation, thereby making better use of limited computational resources. 
    
    Moreover, as FLOPs increase beyond a certain threshold (1.065M FLOPs for MNIST and 12.190M FLOPs for CIFAR10), the performance of the two RDC-based models converges, with both approaching a similar saturation level. This behavior aligns with the trends observed in Table~\ref{tab:udim-flops-info}, where the mutual information $I(X;U)$ levels off once the representation capacity becomes sufficient to capture the essential semantic information. Together, these findings validate that the complexity term $I(X;U)$ not only provides a theoretical tradeoff dimension but also guides the design of more efficient encoders in practice, particularly under strict computational constraints typical of IoT and edge devices.
    To summarize, training semantic coders with a model complexity constraint facilitates more effective capture of source information, leading to improved task performance under the same computational budget when compared to models trained without such a constraint.

    \section{Conclusion}
    \label{sec:conclusion}

    In this paper, we have investigated the fundamental tradeoff among transmission rate, distortion and model complexity. We have considered a semantic communication system where the encoder can only access the semantic source through indirect observations, while both the encoder and the decoder can utilize side information. We have derived the closed-form expressions of the RDC functions for both Gaussian and binary semantic sources, which reveal not only the corresponding three-way tradeoff, but also the fact that communication resources can be offloaded by increasing model complexity of the DNN-based encoders at the user end. Experimental results on practical image data sources verify the theoretical tradeoff and further demonstrate that the proposed complexity measure effectively bridges information-theoretic analysis and practical computational costs, providing a principled guideline for balancing communication and computation resources in semantic communication systems.


    \makeatletter
	\newcommand{\myappendix}{%
		\setcounter{section}{0}
		\renewcommand{\thesection}{Appendix~\Alph{section}}
	}
	\makeatother
	\myappendix

    \section{
		Proof of Theorem \ref{th:gaus}}
	\label{app:proof_gaus}

    Since the mutual information between two Gaussian distributions is invariant to their mean parameters, without loss of generality, we may assume zero means for both  $\hat{S}$ and $S$. For mathematical tractability, we thus model $\hat{S}$ as a zero-mean Gaussian random variable, i.e., $\hat{S} \sim \mathcal{N}(0, \sigma^2)$.
    When the equation of the constraint on complexity holds, the parameter $\rho$ is determined by
    \begin{equation}
        \cc=I(X;U)=\frac12\log\frac1{1-\rho^2},
    \end{equation}
    which leads to $\rho=\sqrt{1-2^{-2\cc}}$. 
    Since the Gaussian RDC problem is in fact an indirect source coding problem,
    we then turn to convert the RDC problem into an equivalent direct source coding form.
    By definitions \eqref{eq:gaus_assumption1} and \eqref{eq:gaus_assumption2}
    we have
    $\mathrm{Cov}(S,X)=\gamma$, and
    \begin{equation}
        S = \gamma\rho U + \sqrt{1-\gamma^2}Z_1 + \gamma\sqrt{1-\rho^2}Z_2 ,
    \end{equation}
    hence $\mathrm{Cov}(S,U)=\gamma\rho$.
    Define $\kappa:={\rm Cov}(U,\hat S)$.
    From the Markov chain $S \to X \to U \to \hat{S}$,
    \begin{align}
       & \mathrm{Cov}(S,\hat{S}) = \mathbb{E}[\mathbb{E}[S|U]\mathbb{E}[\hat{S}|U]]
    = \mathbb{E}[\gamma\rho U\kappa U] = \gamma\rho\kappa,\\
    &\mathrm{Cov}(X,\hat{S}) = \mathbb{E}[\mathbb{E}[X|U]\mathbb{E}[\hat{S}|U]]
    = \mathbb{E}[\rho U\kappa U] = \rho\kappa.
    \end{align}
    The covariance matrix is therefore
    \begin{equation}
        {\rm Cov}(S,X,U,\hat S) =
    \begin{bmatrix}
    1 & \gamma & \gamma\rho & \gamma\rho\kappa \\
    \gamma & 1 & \rho & \rho\kappa \\
    \gamma\rho & \rho & 1 & \kappa \\
    \gamma\rho\kappa & \rho\kappa & \kappa & \sigma^2
    \end{bmatrix}.
    \end{equation}
    Define the distortions $d_S = \mathbb{E}[(S-\hat{S})^2]$
    and $d_U = \mathbb{E}[(U-\hat{S})^2]$.
    We then have
    \begin{equation}
    \label{eq:app_ds1}
         d_S = \mathbb{E}[S^2] + \mathbb{E}[\hat{S}^2] - 2\mathbb{E}[S\hat{S}]
         = 1 + \sigma^2 - 2\gamma\rho\kappa
    \end{equation}
    Solving for $\kappa$ gives
    \begin{equation}
    \label{eq:app_kappa}
        \kappa = \frac{1 + \sigma^2 - d_S}{2\gamma\rho}.
    \end{equation}
    Next, compute $d_U$:
    \begin{equation}
    \label{eq:app_du1}
        d_U = \mathbb{\mathbb{E}}[U^2] + \mathbb{E}[\hat{S}^2] - 2\mathbb{E}[U\hat{S}]
         = 1 + \sigma^2 - 2\kappa.
    \end{equation}
   Introducing \eqref{eq:app_ds1} into \eqref{eq:app_du1}  yields
   \begin{align}
       d_U =& 1 + \sigma^2 - 2 \frac{1 + \sigma^2 - d_S}{2\gamma\rho}\nonumber\\
         = &\frac{d_S}{\gamma\rho} + (1 + \sigma^2)(1 - \frac{1}{\gamma\rho}).
         \label{eq:app_du2}
   \end{align}
   Based on \eqref{eq:app_du2}, the constraint $d_S\leq \cd$ is equivalent to
   \begin{equation}
   \label{eq:app_du3}
       d_U\leq \frac{\cd}{\gamma\rho} -  \frac{1-\gamma\rho}{\gamma\rho}(1 + \sigma^2):=\theta_u.
   \end{equation}
    Combining \eqref{eq:app_du3} and the fact that $p_S=p_U$,
    the original Gaussian RDC problem is equivalent to 
    \begin{align}
    \label{eq:problem_gaus}
        &\min I(U;\hat S),\\
        \text{ s.t. } \mathbb{E}[d(U,&\hat S)]\leq \cu,\; 
        d_W(p_U,p_{\hat S})\leq \cp. \nonumber
    \end{align}
    When the perception constraint is inactive, i.e.,
    \begin{equation}
    \label{eq:temp_2}
        d_{W}(p_U,p_{\hat S})=(1-\sigma)^2< \cp.
    \end{equation}
    In this case, the RDC function is equivalent to the traditional Shannon rate-distortion function, written as 
    \begin{equation}
        R^{\mathcal G}(\cc,\cd,\cp)=\frac12\log\frac{1}{\cu},
    \end{equation}
    with $\sigma=1-\cu$.
    Then we have
    Combining (\ref{eq:temp_2})
    \begin{align}
        &\sigma^2_0 < \sigma^2\leq 1\Leftrightarrow
        0\leq \dx < 1-\sigma_0^2\nonumber\\
       \Leftrightarrow& (1+\sigma_0^2)(1-\gamma\rho)\leq \cd < 1+\sigma_0^2(1-2\gamma\rho),
    \end{align}
    where $\sigma_0=1-\sqrt{\cp}$.
    When the perception constraint is active, i.e., 
    $d_{W}(p_U,p_{\hat S})\geq (1-\sigma)^2= \cp,$
    we have $\sigma=\sigma_0=1-\sqrt{\cp}.$
    Taking $d_S=\cd$ in \eqref{eq:app_kappa}, 
    the rate term  is thus given by
    \begin{equation}
        I(U;\hat S)=\frac12\log \frac{\sigma^2}{\sigma^2-\kappa^2}
        =\frac12\log\frac1{1-\left(\frac{
        1+\sigma^2-\cd}{2\gamma\rho\sigma}\right)^2}
    \end{equation}
    which concludes the proof of Theorem \ref{th:gaus}.

    \section{
		Proof of Theorem \ref{th:bin}}
	\label{app:proof_th1}
    \emph{Proof}:
    Define the  complexity-rate-distortion problem
    \begin{align}
         &\min_{p_{UX\hat S}} d_\mathrm{KL}(p_{S|X} \Vert p_{S|\hat S})\nonumber\\
         \text{s.t.}\;&  I(X;\hat S)\leq R, I(X;U)\leq \cc. 
    \end{align}
    The optimal coding scheme, which induces the optimal conditional probabilities $p_{U|X},p_{\hat S|U}$ of the distortion-rate-complexity problem, is identical to that of the RDC problem.
    Thus, we turn to derive the solution to the distortion-rate-complexity problem.
    Firstly, the distortion term can be reformulated as
    $d_\mathrm{KL}(p_{S|X},p_{S|\hat S})=H(S|\hat S) - H_b(q_{SX}).$
	Denote by $q_{SU},q_{S\hat S}$ the crossover probability between $S$ and $U$, and that between $S$ and $\hat S$ respectively.
	We then have
    \begin{align}
    q_{SU}=&q_{SX}+q_{XU}-2q_{SX}q_{XU}\\
        q_{S\hat S}
		=&q_{SX}+q_{XU}+q_{U\hat S}+4q_{SX}q_{XU}q_{U\hat S}\nonumber\\
        &-2(q_{SX}q_{XU}+q_{SX}q_{U\hat S}+q_{XU}q_{U\hat S}),
        \label{eq:q_SShat}
    \end{align}
	and the mutual information of the distortion term is 
	\begin{equation}
        \label{eq:bin_mi_gen}
		I(\hat S;S)=1-H_b(q_{S\hat S}).
	\end{equation}
    Based on (\ref{eq:bin_mi_gen}), the distortion term is written as
    \begin{equation}
        d_\mathrm{KL}(p_{S|X} \Vert p_{S|\hat S})
        =H_b(q_{S\hat S})-H_b(q_{SX}).
    \end{equation}
    %
    We then parameterize the following conditional probabilities:
    \begin{align}
        &p_0:=p(U=1|X=0),\ p_1:=p(U=0|X=1),\\
        &q_0:=p(\hat S=1|U=0),\;\ q_1:=p(\hat S=0|U=1),
    \end{align}
    and the mutual information term of complexity is written as
    \begin{align}
        I(X;U)=&\frac12\bigg( H_b(\frac{p_0+1-p_1}{2}) + H_b(\frac{p_1+1-p_0}{2})\nonumber\\
        &-H_b(p_0) -H_b(p_1) \big)
    \end{align}
    Similarly, the rate term is written as
    \begin{align}
        I(U;\hat S)=&\frac12\bigg( H_b(\frac{q_0+1-q_1}{2}) + H_b(\frac{q_1+1-q_0}{2})\nonumber\\ 
        &-H_b(q_0) -H_b(q_1)\big).
    \end{align} 
    For the distortion term in (\ref{eq:q_SShat}),
    we have
    \begin{align}
        q_{S\hat S}=&q_{SX}+q_{XU}+q_{U\hat S}+4q_{SX}q_{XU}q_{U\hat S}\nonumber\\
        &-2(q_{SX}q_{XU}+q_{SX}q_{U\hat S}+q_{XU}q_{U\hat S}),
    \end{align}
    where $q_{XU}=\frac{p_0+p_1}2$, $q_{U\hat S}=\frac{q_0+q_1}2$.
    Then the binary RDC problem can be reformulated as:
    \begin{align}
    \label{eq:minimization}
		&\min_{p_0,p_1,q_0,q_1\leq 1} H_b(q_{S\hat S}),\\
        \text{s.t. }&I(X;U)\leq \cc, I(U;\hat S)\leq R\nonumber
	\end{align}
    We associate the following Lagrangian function to solve the optimization problem in (\ref{eq:minimization}):
    \begin{align}
		\mathcal L=& H_b(q_{S\hat S})+\lambda_1(I(X;U)-\cc)
        +\lambda_2(I(U;\hat S)-R)\nonumber\\
		&\hspace{-28pt}+\lambda_3(p_0-1)+\lambda_4(q_1-1)
		+\lambda_5(q_0-1)+\lambda_5(q_1-1).
	\end{align}
    For $0<p_0,p_1,q_0,q_1 < 1$,
		we have $\lambda_3=\lambda_4=\lambda_5=\lambda_6=0$.
		Then for we have the following equations:
        \begin{equation}
        \label{eq:app1_pd_1}
            \frac{\partial H_b(q_{S\hat S})}{\partial p_0} 
            +\lambda_1 \frac{\partial}{\partial p_0}
            \left( H_b(\frac{p_0+1-p_1}{2}) -\frac12H_b(p_0)\right) = 0,
        \end{equation}
        \begin{equation}
        \label{eq:app1_pd_2}
            \frac{\partial H_b(q_{S\hat S})}{\partial p_1} 
            +\lambda_1 \frac{\partial}{\partial p_1}
            \left( H_b(\frac{p_1+1-p_0}{2}) -\frac12H_b(p_1)\right) = 0,
        \end{equation}
        \begin{equation}
        \label{eq:app1_pd_3}
             \frac{\partial H_b(q_{S\hat S})}{\partial q_0} 
            +\lambda_2 \frac{\partial}{\partial q_0}
            \left( H_b(\frac{q_0+1-q_1}{2}) -\frac12H_b(q_0)\right) = 0,
        \end{equation}
        \begin{equation}
        \label{eq:app1_pd_4}
            \frac{\partial H_b(q_{S\hat S})}{\partial q_1} 
            +\lambda_2 \frac{\partial}{\partial q_1}
            \left( H_b(\frac{q_1+1-q_0}{2}) -\frac12H_b(q_1)\right) = 0,
        \end{equation}
        where (\ref{eq:app1_pd_1}), (\ref{eq:app1_pd_2}), (\ref{eq:app1_pd_3}), (\ref{eq:app1_pd_4}) 
        are the partial derivations 
        $\frac{\partial \mathcal L}{\partial p_0}$, $\frac{\partial \mathcal L}{\partial p_1}$, $\frac{\partial \mathcal L}{\partial q_0}$, $\frac{\partial \mathcal L}{\partial q_1}$ respectively,
        leading to the solutions:
    \begin{equation}
        p_0=p_1=q_{XU},\; q_0=q_1=q_{U\hat S},
    \end{equation} 
    where $q_{XU}$ and $q_{U\hat S}$ are determined by
    \begin{equation}
        \cc=1-H_b(q_{XU}),\;
        \cp=1-H_b(q_{U\hat S}).
    \end{equation}
    Hence, the optimal coding scheme for the distortion-rate-complexity problem uses doubly symmetric binary channels (DSBC): let $p_{U|X} \sim \text{DSBC}(q_{XU})$ and $p_{\hat S|U} \sim \text{DSBC}(q_{U\hat S})$, leading to
    \begin{equation}
        d_\mathrm{KL}(p_{S|X} \Vert p_{S|\hat S})
        =H_b(q_{S\hat S})-H_b(q_{SX}).
        \label{eq:bin_distortion}
    \end{equation}
    Since the distortion-rate-complexity problem is equivalent to the RDC problem, the proposed coding scheme is optimal and achieves the RDC function. Thus, (\ref{eq:bin_distortion}) is precisely the RDC function.
    This concludes the proof of Theorem \ref{th:bin}.
    \hfill $\blacksquare$

\bibliographystyle{IEEEtran}
\bibliography{cjx}

\begin{IEEEbiography}[{\includegraphics[width=1in,height=1.25in,clip,keepaspectratio]{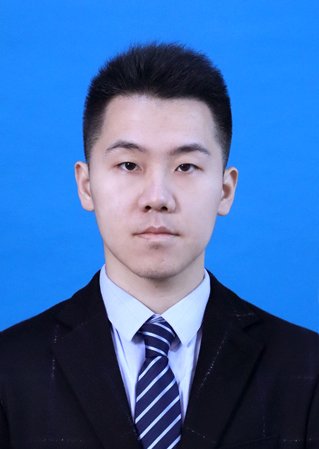}}]{Jingxuan Chai}
Jingxuan Chai received his B.E. degree in communication engineering from Xidian University, Xi'an, China, in 2019. 
He is currently pursuing the Ph.D. degree with the School of Artificial Intelligence, Xidian University.
His research interests include information theory, semantic communications, agentic communications and representation learning.
\end{IEEEbiography}

\begin{IEEEbiography}[{\includegraphics[width=1in,height=1.25in,clip,keepaspectratio]{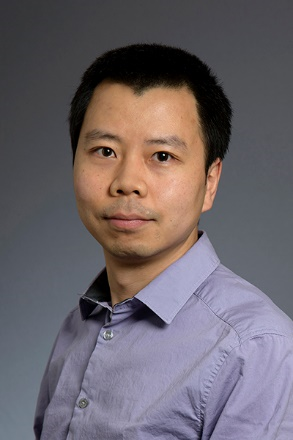}}]{Yong Xiao}
Yong Xiao (Senior Member, IEEE) received his B.S. degree in electrical engineering from China University of Geosciences, Wuhan, China, in 2002, M.Sc. degree in telecommunication from Hong Kong University of Science and Technology in 2006, and his Ph.D degree in electrical and electronic engineering from Nanyang Technological University, Singapore, in 2012. He is now a professor in the School of Electronic Information and Communications at the Huazhong University of Science and Technology (HUST), Wuhan, China. He is also with Peng Cheng Laboratory, Shenzhen, China, and Pazhou Laboratory (Huangpu), Guangzhou, China. He is the associate group leader of the Network Intelligence Group of IMT-2030 (6G promoting group) and the Vice Director of the 5G Verticals Innovation Laboratory at HUST. Before he joined HUST, he was a research assistant professor in the Department of Electrical and Computer Engineering at the University of Arizona, where he was also the center manager of the Broadband Wireless Access and Applications Center (BWAC), an NSF Industry/University Cooperative Research Center (I/UCRC) led by the University of Arizona. His research interests include machine learning, game theory, distributed optimization, and their applications in semantic communications, semantic-aware networking, cloud/fog/mobile edge computing, green communication systems, and the Internet-of-Things (IoT).
\end{IEEEbiography}

\begin{IEEEbiography}[{\includegraphics[width=1.1in,height=1.3in, clip,keepaspectratio]{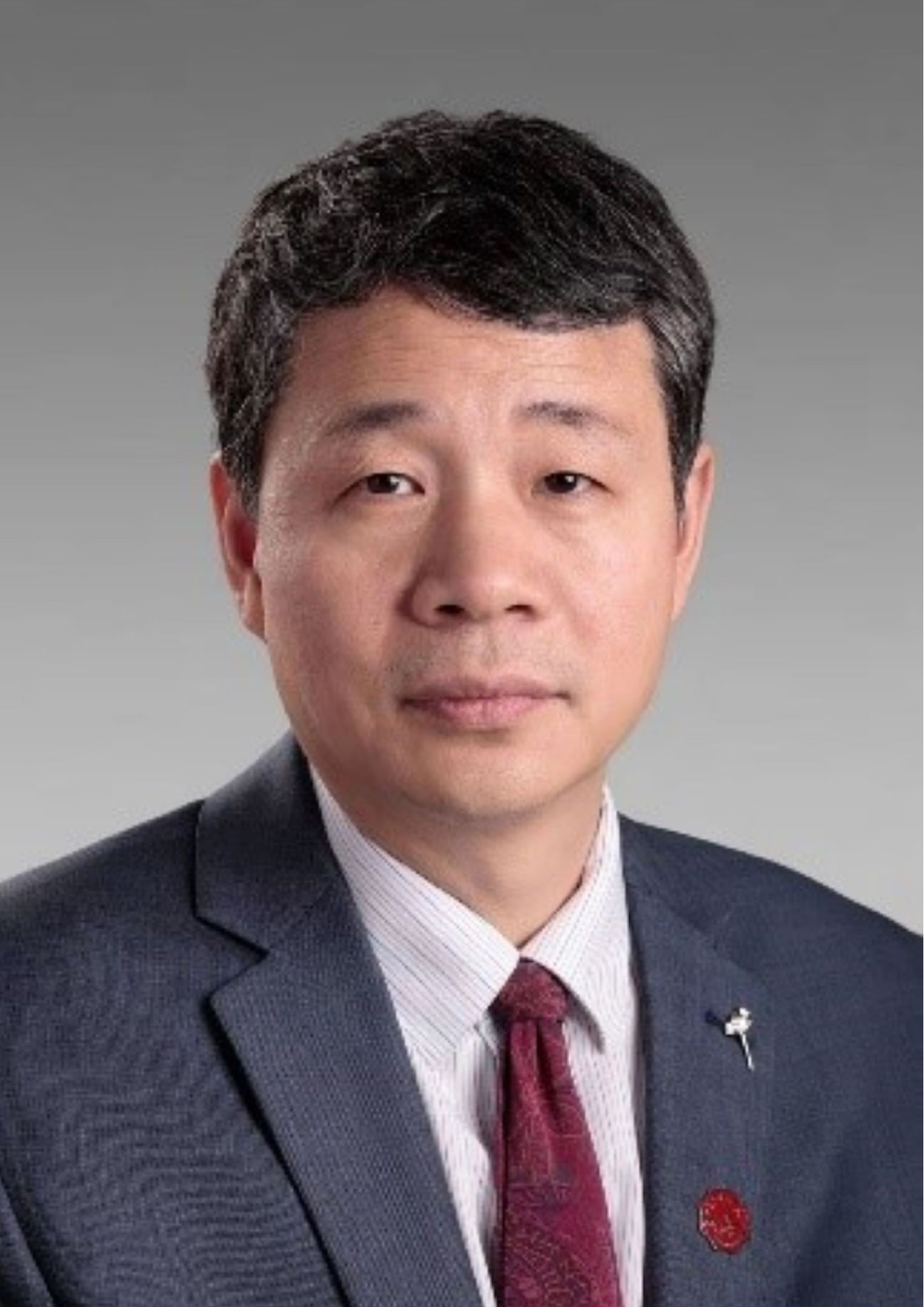}}]{Guangming Shi} (Fellow, IEEE) 
is the Vice Dean of Peng Cheng Laboratory and a Professor with the School of Artificial Intelligence, Xidian University. He is an IEEE Fellow, the chair of IEEE CASS Xi’an Chapter, a senior member of ACM and CCF, Fellow of the Chinese Institute of Electronics, and Fellow of IET. 
His research interests include artificial intelligence, semantic communications, and human-computer interaction.
\end{IEEEbiography}

\end{document}